\documentclass[oribibl]{llncs}
\usepackage{amssymb,amsmath}
\usepackage{graphicx,latexsym,float}
\usepackage{algorithm}
\usepackage[noend]{algorithmic}
\usepackage{versions}
\usepackage[square,numbers]{natbib}
\usepackage{thm-restate}
\usepackage[usenames,dvipsnames]{color}
\usepackage[dvipsnames,svgnames,table]{xcolor}
\usepackage{xspace}
\usepackage[hyphens]{url}

\newcommand{\R}{{\cal R}}
\newcommand{\B}{{\cal B}}
\newcommand{\prefixsumarray}{{\cal A}}
\newcommand{\AC}{{\sf AC}}
\newcommand{\NP}{{\sf NP}}
\newcommand{\fsram}{{\sc fs-ram}}  
\newcommand{\fsrasmall}{fs-ram}
\newcommand{\fsrambig}{FS-RAM}
\newcommand{\ram}{{\sc ram}}
\newcommand{\wordRAM}{word-\ram}
\newcommand{\uwram}{{\sc uw-ram}}

\newcommand{\comp}{{compress}}
\newcommand{\sprd}{{spread}}
\newcommand{\alu}{{\sc alu}\xspace}
\newcommand{\bmh}{{\sc bmh}\xspace}

\makeatletter
\renewcommand\bibsection%
{
  \section*{\refname
    \@mkboth{\MakeUppercase{\refname}}{\MakeUppercase{\refname}}}
}
\makeatother

\includeversion{FULL}
\excludeversion{SHORT}
\excludeversion{OLD}

\begin{document}

\title{Algorithms in the Ultra-Wide Word Model}
\titlerunning{Algorithms in the Ultra-Wide Word Model}  
\author{Arash Farzan\inst{1} \and Alejandro L\'{o}pez-Ortiz\inst{2} \and Patrick K. Nicholson\inst{3} \and Alejandro Salinger\inst{4}
}
\authorrunning{Farzan et al.}
\tocauthor{Arash Farzan,Alejandro L\'{o}pez-Ortiz,Patrick K. Nicholson,Alejandro Salinger}

\institute{Facebook Inc.\\
\email{afarzan@fb.com}
\and 
David R. Cheriton School of Computer Science, University of Waterloo\\
\email{alopez-o@uwaterloo.ca}
\and
Max-Planck-Institut f\"{u}r Informatik\\
\email{pnichols@mpi-inf.mpg.de}
\and 
Department of Computer Science, Saarland University \\
\email{salinger@cs.uni-saarland.de}
}
\maketitle      
\begin{abstract}
The effective use of parallel computing resources to speed up
algorithms in current multi-core parallel architectures
remains a difficult challenge, with ease of programming playing a key
role in the eventual success of various parallel architectures.
In this paper we consider an alternative view of parallelism in the
form of an ultra-wide word processor.
We introduce the Ultra-Wide Word architecture and model, an extension
of the word-\ram~model that allows for constant time operations on
thousands of bits in parallel.
Word parallelism as exploited by the word-\ram~model does not suffer
from the more difficult aspects of parallel programming, namely
synchronization and concurrency.
%
For the standard word-\ram~algorithms, the speedups obtained are
moderate, as they are limited by the word size.
We argue that a large class of word-\ram~algorithms can be implemented
in the Ultra-Wide Word model, obtaining speedups comparable to
multi-threaded computations while keeping the simplicity of
programming of the sequential \ram~model.
We show that this is the case by describing implementations of
Ultra-Wide Word algorithms for dynamic programming and string
searching.
In addition, we show that the Ultra-Wide Word model can be used to
implement a non-standard memory architecture, which enables the
sidestepping of lower bounds of important data structure problems such
as priority queues and dynamic prefix sums.
While similar ideas about operating on large words have been mentioned before in the context of multimedia processors~\cite{Thorup03}, it is only recently that an  architecture like the one we propose has become feasible and that details can be worked out.
\end{abstract}

%
%
\section{Introduction}
In the last few years, multi-core architectures have become the
dominant commercial hardware platform. The potential of these architectures to
improve performance through parallelism remains to be fully attained, as
effectively using all cores on a single application has proven to be a
difficult challenge. In this paper we introduce the Ultra-Wide Word architecture and model of computation, an alternate view of parallelism for
a modern architecture in the form of an ultra-wide word
processor. This can be implemented by replacing one or more cores of a multi-core chip with
a very wide word Arithmetic Logic Unit (\alu) that can perform
operations on a very large number of bits in parallel.

The idea of executing operations on a large number of bits
simultaneously has been successfully exploited in different forms. In
Very Long Instruction Word (VLIW) architectures~\cite{Fisher:1983:VLI:1067651.801649}, several instructions can be
encoded in one wide word and executed in one single parallel
instruction. Vector processors allow the execution of one instruction on
multiple elements simultaneously, implementing
Single-Instruction-Multiple-Data (SIMD) parallelism. This form of
parallelism led to the design of supercomputers such as the Cray
architecture family~\cite{cray1} and is now present in Graphics
Processing Units (GPUs) as well as in Streaming SIMD Extensions (SSE)
 to scalar processors.

In 2003, Thorup~\cite{Thorup03} observed that certain instructions present in some SSE implementations were particularly useful for operating on large integers and speeding up algorithms for combinatorial problems. To a certain extent, some of the ideas in the Ultra Wide Word architecture are presaged in the paper by Thorup, which was
proposed in the context of multimedia processors. Our architecture developed
independently and differs on several aspects 
\begin{SHORT}
(see discussion in full version~\cite{uwram_full})
\end{SHORT}
\begin{FULL}
(see discussion in Section~\ref{sec:other_models})
\end{FULL} 
but it is motivated by similar considerations.

As CPU hardware advances, so does the model used in theory to analyze
it. The increase in word size was reflected in the word-\ram~model in
which algorithm performance is given as a function of the input size
$n$ and the word size $w$, with the common assumption that
$w=\Theta(\log n)$. In its simplest version, the word-\ram~model allows
the same operations as the traditional \ram~model.  
Algorithms in this
model take advantage of bit-level parallelism through packing various
elements in one word and operating on them simultaneously. Although
similar to vector processing, the word-\ram~provides more flexibility
in that the layout of data in a word depends on the algorithm and
data elements can be packed in an arbitrary way. Unlike VLIW
architectures, the Ultra-Wide Word model we propose is not concerned
with the compiler identifying operations which can be done in parallel
but rather with achieving large speedups in implementations of
word-\ram~algorithms through operations on thousands of bits in
parallel.\looseness=-1

As multi-core chip designs evolve, chip vendors try to determine the
best way to use the available area on the chip, and the options
traditionally are an increased number of cores or larger caches.  We believe that the
current stage in processor design allows for the inclusion of an
architecture such as the one we propose. In addition,
ease of programming is a major hurdle to the eventual success of
parallel and multi-core architectures. In contrast, bit parallelism as
exploited by the word-\ram~model does not suffer from this drawback:
there is a large selection of word-\ram~algorithms (see,
e.g., \cite{andersson07,Han:2004:DSO:975978.975984,hagerup98,C06}) that readily
benefit from bit parallelism without having to deal with the more
difficult aspects of concurrency such as mutual exclusion,
synchronization, and resource contention. In this sense, the advantage
of an on-chip ultra-wide word architecture is that it can enable
word-\ram~algorithms to achieve speedups comparable to those of multi-threaded
computations, while at the same time keeping the simplicity of
sequential programming that is inherent to the \ram~model. We argue
that this is the case by showing several examples of implementations
of word-\ram~algorithms using the wide word, usually with simple
modifications to existing algorithms, and extending the ideas and
techniques from the word-\ram~model.\looseness=-1


In terms of the actual architecture, we envision the ultra-wide \alu
together with multi-cores on the same chip. Thus, the Ultra-Wide Word
architecture adds to the computing power of current architectures. The
results we present in this paper, however, do not use multi-core parallelism.

\begin{SHORT}
\paragraph{\textbf{Summary of Results}}
We introduce the Ultra-Wide Word architecture and
\end{SHORT}
\begin{FULL}
\subsubsection{Summary of Results}
We introduce the Ultra-Wide Word architecture and
\end{FULL}
model, which extends the $w$-bit word-\ram\ model by adding an \alu that
operates on $w^2$-bit words.  We show that several broad classes of
algorithms can be implemented in this model.  In particular:

\begin{itemize}

\item We describe Ultra-Wide Word implementations of dynamic
  programming algorithms for the subset sum problem, the knapsack
  problem, the longest common subsequence problem, as well as many
  generalizations of these problems.  Each of these algorithms
  illustrates a different technique (or combination of techniques) for
  translating an implementation of an algorithm in the word-\ram~model
  to the Ultra-Wide Word model.  In all these cases we obtain a
  $w$-fold speedup over word-\ram~algorithms.

\item We also describe Ultra-Wide Word implementations of popular
  string searching algorithms: the Shift-And/Shift-Or
  algorithms~\cite{by92,wu92} and the Boyer-Moore-Horspool
  algorithm~\cite{horspool80}. Again, we obtain a $w$-fold speedup over the original
  algorithms.

\item Finally, we show that the Ultra-Wide Word model is powerful
  enough to simulate a non-standard memory architecture in which bytes can overlap, which we shall call \fsram~\cite{fredman89}. This allows us to
  implement data structures and algorithms that circumvent known lower
  bounds for the word-\ram~model.
  
\end{itemize}

\begin{SHORT}
Due to space constraints, we only present a high-level description of our results. The full details can be found in the full version of this paper~\cite{uwram_full}.\looseness=-1
\end{SHORT}

\begin{FULL}
The rest of this paper is organized as follows. In
Section~\ref{sec:uwRAM} we describe the Ultra-Wide architecture and
model of computation. We show in Section~\ref{sec:rambo} how to
simulate the \fsram~memory
architecture. In Sections~\ref{sec:dp} and~\ref{sec:string_search} we
present \uwram~implementations of algorithms for dynamic
programming and string searching. 
We present concluding remarks in Section~\ref{sec:conclusions}.
\end{FULL}

\section{The Ultra-Wide Word-RAM Model}
\label{sec:uwRAM}

The Ultra-Wide word-\ram~model (\uwram) we propose is an extension of the word-\ram~model. 
\begin{FULL} 
We briefly review here the key features of the word-\ram.

\subsection{Algorithms in the word-RAM model}
\label{sec:wordRAM}
\end{FULL}
The word-\ram~is a variant of the \ram~model in which a word has length $w$ bits, and the contents of memory are integers in the range $\{0,\ldots,2^{w}-1\}$~\cite{hagerup98}. This implies that $w\ge \log n$, where $n$ is the size of the input, and a common assumption is $w=\Theta(\log n)$ (see, e.g.,~\cite{munro96,Bose09}). 
\begin{SHORT}
Algorithms in this model take advantage of the intrinsic parallelism of operations
on $w$-bit words. We provide a more detailed description of the word-\ram~in the full version~\cite{uwram_full}.\looseness=-1
\end{SHORT}
\begin{FULL}
The word-\ram~includes the usual load, store, and jump instructions of the
\ram~model, allowing for immediate operands and for direct and indirect
addressing. In this model, arithmetic operations on two words are
modulo $2^w$, and the instruction set includes left and right shift
operations (equal to multiplication and division by powers of two) and
boolean operations. All instructions take constant time to
execute. There are different versions of the word-\ram~model depending
on the instruction set assumed to be available. The \emph{restricted
  model} is limited to addition, subtraction, left and right shifts,
and boolean operations AND, OR, and NOT. These instructions augmented
with multiplication constitute the \emph{multiplication
  model}. Finally, the $\AC^0$ model assumes that all functions
computable by an unbounded fan-in circuit of polynomial size (in $w$)
and constant depth are available in the instruction set and execute in
constant time. This definition includes all instructions from the
restricted model and excludes multiplication. We refer to the reader
to the survey by \citet{hagerup98} for a more extended
description of the model and a discussion of its practicality.

Word-\ram~algorithms
exploit word-level parallelism by operating on various elements simultaneously
using instructions on $w$-bits words. There are various algorithms for fundamental problems that take
advantage of word-level parallelism or a bounded universe, some of
which fit into the word-\ram~model, although are not explicitly
designed for it~\cite{fourRussians70}. Much attention has been given
to sorting and searching, for which known lower bounds in the
comparison model do not carry to the word-\ram~model~\cite{FW93}. For
example, in a word-\ram~model with multiplication, sorting $n$ words
can be done in $O(n\log\log n)$ time and $O(n)$ space
deterministically~\cite{Han:2004:DSO:975978.975984}, and in expected
$O(n\sqrt{\log\log n})$ time and $O(n)$ space using
randomization~\cite{han02}.  Word-\ram~techniques have also been
applied in many different areas, such as succinct data
structures~\cite{J89,munro96}, computational geometry~\cite{C06,CP09},
and text indexing~\cite{GGV03}.

\subsection{Ultra-Wide RAM}
\end{FULL}

The Ultra-Wide word-\ram~model 
\begin{FULL}
(\uwram)
\end{FULL} 
extends the word-\ram~model by
introducing an ultra-wide \alu with $w^2$-bit \emph{wide words}\begin{FULL}, where
$w$ is the number of bits in a word-\ram\end{FULL}. The ultra-wide \alu supports
the basic operations available in a word-\ram~on the entire word at once. 
 As in the word-\ram~model, the available set of instructions can be assumed to be those of the restricted, multiplication, or the $\AC^0$ models. For the results in this paper we assume the instructions of the restricted model (addition, subtraction, left and right shift, and bitwise boolean
operations), plus two non-standard straightforward $\AC^0$
operations that we describe at the end of this \begin{FULL}sub\end{FULL}section.

The model maintains the standard $w$-bit \alu as well as
$w$-bit memory addressing. In general, we use the parameter $w$ for the word size in the description and analysis of algorithms, although in some cases we explicitly assume $w=\Theta(\log n)$. In terms of real world parameters, the wide
word in the ultra-wide \alu would presently have between 1,000 and
10,000 bits and could increase even further in the future. 
In reality, the addition of an \alu that supports operations on thousands on bits would require appropriate adjustments to the data and instruction caches of a processor as well as to the instruction pipeline implementation. Similarly to the abstractions made by the \ram~and word-\ram~models, the \uwram~model ignores the effects of these and other architectural features and assumes that the execution of instructions on ultra-wide words is as efficient as the execution of operations on regular $w$-bit words, up to constant factors.

Provided that the \uwram~supports the same operations as the word-\ram,
the techniques to achieve bit-level parallelism in the word-\ram~extend
directly to the \uwram. However, since the word-\ram~assumes that a
word can be read from memory in constant time, many operations in
word-\ram~algorithms can be implemented 
 through constant time table lookups. 
\begin{FULL} 
For
example, counting the number of set bits in a word of $w=\log n$ bits
can be implemented through two table lookups to a precomputed table
that stores the number of set bits for each number of $\log n/2$
bits. The space used by the table is $\sqrt{n}$ words.
We cannot expect to achieve the same constant time lookup operation with words
of $w^2$ bits since the size of the lookup tables would be
prohibitive. \end{FULL}
\begin{SHORT}
With words of $w^2$ bits, we cannot expect to achieve constant time lookups since the size of the tables would be
prohibitive.
\end{SHORT}
However, the memory access operations of our model allow for the implementation of simultaneous table lookups of several $w$-bit words within a wide word, as we shall explain below.

We first introduce some notation. Let $W$ denote a $w^2$-bit word. Let
$W[i]$ denote the $i$-th bit of $W$, and let $W[i..j]$ denote the
contiguous subword of $W$ from bit $i$ to bit $j$, inclusive. The
least significant bit of $W$ is $W[0]$, and thus
$W=\sum_{i=0}^{w^2-1} W[i]\times 2^i$. 
For the sake of memory access operations, 
we divide $W$ into $w$-bit blocks. Let $W_j$ denote the $j$-th contiguous block
of $w$ bits in $W$, for $0\le j \le w-1$, and let $W_j[i]$ denote the
$i$-th bit within $W_j$. Thus, $W_j=W[jw..(j+1)w-1]$\begin{FULL} and $W=\sum_{j=0}^{w-1}2^{jw}\times(\sum_{i=0}^{w-1} W_j[i]\times2^i)$\end{FULL}. 
The division of a wide word in blocks is solely intended for certain
memory access operations, but basic operations of the model have no notion of
block boundaries. Fig.~\ref{fig:word} shows a representation of a
wide word, depicting bits with increasing significance from left
to right. In the description of operations
with wide words we generally refer to variables with uppercase
letters, whereas we use lowercase to refer to regular variables that use one $w$-bit word. Thus, shifts to the left (right) by $i$ are equivalent to
division (multiplication) by $2^i$.  In addition, we use $\vec{0}$ to denote a wide word with value
0. We use standard C-like notation for operations {\sc and} (`$\&$'), {\sc or}
(`$|$'), {\sc not} (`$\sim$') and shifts (`$<<$',`$>>$').\looseness=-1


 \begin{figure*}[!t]
\begin{center}
\begin{FULL}
\includegraphics[scale=0.60]{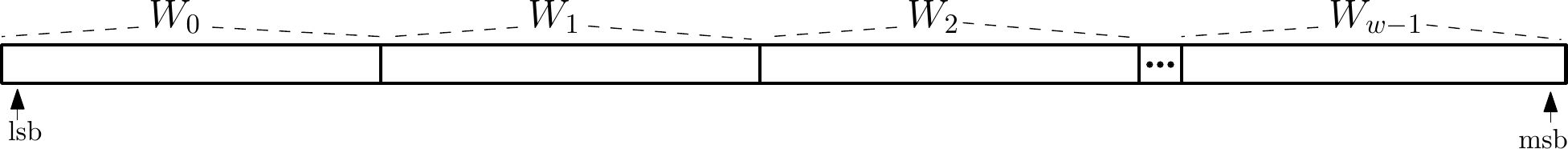}
\end{FULL}
\begin{SHORT}
\includegraphics[scale=0.50]{figs/wordram.pdf}
\end{SHORT}
\caption[A wide word in the Ultra-Wide Word architecture]{A wide word in the Ultra-Wide Word architecture. The wide
  word is divided in $w$ blocks of $w$ bits each, shown here in
  increasing number of block from left to right.}
\label{fig:word}
\end{center}
\end{figure*}

\begin{SHORT}
\paragraph{\textbf{Memory access operations}}  
In this architecture $w$ (not necessarily contiguous) words from
\end{SHORT}
\begin{FULL}
\subsubsection{Memory Access Operations}
In this architecture $w$ (not necessarily contiguous) words from
\end{FULL}
memory can be transferred into the $w$ blocks of a wide word $W$ in constant
time. These blocks can be written to memory in parallel as
well. 
As with PRAM algorithms, the memory access type of the model can be assumed to allow or disallow concurrent reads and writes.  For the results in this paper we assume the Concurrent-Read-Exclusive-Write (CREW) model.

The memory access operations that involve wide words are of three types:
block, word, and content. We describe read accesses (write accesses are analogous). A \emph{block access} loads a single $w$-bit word from memory into a given block of a wide word. A \emph{word access} loads $w$ contiguous $w$-bit words from memory into an entire wide word in constant time. Finally, a \emph{content access} uses the contents of a wide word $W$ as addresses to load (possibly non-contiguous) words of memory simultaneously: for each block $j$ within $W$,  this operation loads from memory the $w$-bit word whose address is $W_j$ (plus possibly a base address).  The specifics of read and write operations are shown in Table~\ref{tab:memory_access}. 

Note that
accessing several (possibly non-contiguous) words from memory
simultaneously is an assumption that is already made by any shared
memory multiprocessing model. While, in reality, simultaneous access
to all addresses in actual physical memory (e.g., DRAM) might not be
possible, in shared memory systems, such as multi-core processors, the
slowdown is mitigated by truly parallel access to private and shared
caches, and thus the assumption is reasonable. We therefore follow
this assumption in the same spirit. 
In fact, for $w$ equal to the regular word size (32 or 64 bits), the choice of $w$ blocks of $w$ bits each for the wide word \alu was judiciously made to provide the model with a feasible memory access implementation. $w^2$ lines to memory are well within the realm of the possible, as they are of the same order of magnitude (a factor of 2 or 8) as modern GPUs, some of which feature bus widths of 512 bits 
\begin{SHORT}(see, e.g.,~\cite{AMDFirePro,NvidiaGTX285}).\end{SHORT}
\begin{FULL}(e.g., FirePro W9100~\cite{AMDFirePro} or Nvidia GeForce GTX 285~\cite{NvidiaGTX285}, see also~\cite{AMD_GPUs,Nvidia_GPUs}).\end{FULL}
We note that a more general model could feature a wide word with $k$ blocks of $w$ bits each, where $k$ is a parameter, which can be adjusted in reality according to the feasibility of implementation of parallel memory accesses. Although described for $w$ blocks, the algorithms presented in this paper can easily be adapted to work with $k$ blocks instead. 
Naturally, the speedups obtained would depend on the number of blocks assumed, but also on the memory bandwidth of the architecture. A practical implementation with a large number of blocks would likely suffer slowdowns due to congestion in the memory bus. We believe that an implementation with $k$ equal to 32 or 64 can be realized with truly parallel memory access, leading to significant speedups.

\begin{table*}[!t]
\begin{center}
\begin{tabular}{| l | l | l |}\hline
\ \textbf{Name}\  & \textbf{Input} &\   \textbf{Semantics}\\\hline\hline
\ read\_block\  &\  $W$, $j$, base\ &\  $W_j \leftarrow $MEM[base+$j$]\  \\\hline 
\ read\_word\  &\  $W$, base\  &\  for all $j$ in parallel: $W_j \leftarrow $MEM[base+j]\ \\\hline
\ read\_content\  &\ $W$, base \  &\  for all $j$ in parallel: $W_j \leftarrow $MEM[base+$W_j$]\ \\\hline
\ write\_block\  &\  $W$, $j$, base\  &\  MEM[base+$j$]$\leftarrow W_j$\ \\\hline
\ write\_word\  &\  $W$, base\  &\  for all $j$ in parallel: MEM[base+$j$]$ \leftarrow W_j$ \ \\\hline
\ write\_content &\ $W$, $V$, base\ &\ for all $j$ in parallel: MEM[base+$V_j$]$ \leftarrow W_j$ \ \\\hline
\end{tabular}
\end{center}
\caption[Wide word memory access operations supported by the
  UW-RAM]{Wide word memory access operations of the
  \uwram. {\sc mem} denotes regular \ram~memory, which is indexed by addresses to
words
, and \emph{base} is some base address.} 
\label{tab:memory_access}
\end{table*}

\begin{SHORT}
\paragraph{\textbf{UW-RAM Subroutines}}
A procedure called \emph{\comp} serves to bring together bits from all blocks
\end{SHORT}
\begin{FULL}
\subsubsection{UW-RAM Subroutines}
\label{sec:utilities}

We now describe some operations that will be used throughout the
\uwram~implementations that we describe in later sections. 
A procedure called \emph{\comp} serves to bring together bits from all blocks
\end{FULL}
into one block in constant time, while a procedure called
\emph{\sprd} is the inverse function\footnote{These operations are also known as PackSignBits and UnPackSignBits~\cite{Thorup03}.}.
Both operations can be implemented by straightforward constant-depth circuits. 
\begin{OLD}Alternatively, if multiplication on the wide word is available, these operations can be implemented by adapting the compression and spreading operations described by Brodnik~\cite[Chapter
  4.]{brodnik95}.
\end{OLD}	
We will also use parallel comparators, a standard technique used in word-\ram~algorithms~\cite{hagerup98} (see details in \begin{FULL}Appendix~\ref{app:subroutines}\end{FULL}\begin{SHORT}full version~\cite{uwram_full}\end{SHORT}).
Although these are all the subroutines that we need for the results in this paper, other operations of similar complexity could be defined if proved useful. 


\begin{figure*}[!t]
\begin{center}
\begin{SHORT}
\includegraphics[scale=0.50]{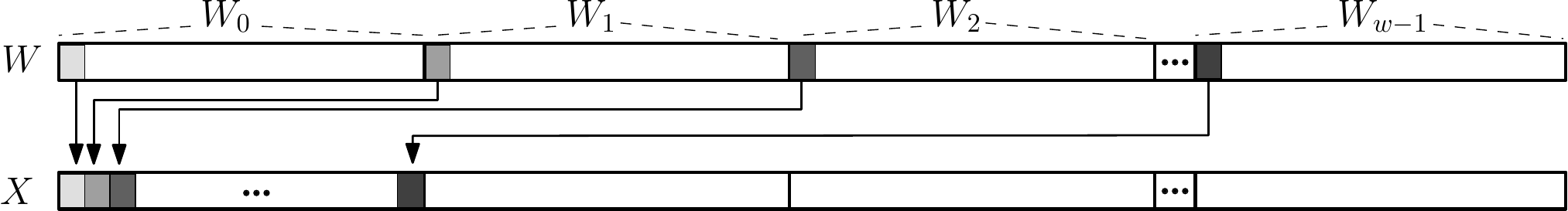}
\end{SHORT}
\begin{FULL}
\includegraphics[scale=0.50]{transpose.pdf}
\end{FULL}
\caption{The \emph{\comp} operation takes a wide word $W$ whose
  set bits are restricted to the first bit of each block and
  compresses them to the first block of a wide word.}
\label{fig:transpose}
\end{center}
\end{figure*}

\begin{itemize}
\item \textbf{Compress}: 
Let $W$ be a wide word in which all bits are zero except possibly for the
first bit of each block. The \comp~operation copies the first bit
of each block of $W$ to the first block of a word $X$. I.e., if
$X=\textrm{\comp}(W)$, then $X[j]\gets W_j[0]$ for $0 \le j < w$, and
$X[j]=0$ for $j\ge w$ (see
Fig.~\ref{fig:transpose}). 
\item \textbf{Spread:} 
This operation is the inverse of the \comp~operation. It takes a
word $W$ whose set bits are all in the first block and spreads them
across blocks of a word $X$ so that $X_j[0]\gets W[j]$ for $0 \le j <
w$. 
\end{itemize}

\begin{SHORT}
\paragraph{\textbf{Relation to Other Models}}
\label{sec:other_models}

We provide a discussion of similarities and differences between the \uwram~and other existing models in the full version~\cite{uwram_full}.
\end{SHORT}

\begin{FULL}
\subsection{Relation to Other Models}
\label{sec:other_models}

There exist various models and architectures that exploit the execution of instructions on a large number of bits simultaneously.  
In Very Large Instruction Word (VLIW) architectures~\cite{Fisher:1983:VLI:1067651.801649} several, possibly different instructions can be encoded in one wide word and executed in parallel. It is usually the compiler's job to determine which instructions of a program can be executed safely in parallel. In contrast, in the \uwram~model it is up to the algorithm designer to specify how parallelism in the ultra wide word should be used. In addition, the wide word can only execute one type of instruction at a time. In this sense, the \uwram~is closer to a vector processor, in which a single instruction is executed on various data item, implementing SIMD parallelism. However, 
while vector processors operate on fields which are independent of each other, the ultra wide \alu in the \uwram~is really one wide word of thousands of bits that treats its contents as one data object. An exception to this are the memory access instructions, which load and store data in blocks within the wide word so that the wide word \alu can interact with regular $w$-bit data. It is of course possible to use the ultra-wide word to implement a vectorized operation, however, as instructions in the \uwram~operate on the entire word, it is up to the algorithm designer to deal with carries and other interference within fields. Moreover, the length of a field in the \uwram~is variable, as it depends on the algorithm's choice. In that sense, the \uwram~is a more flexible model. 

Many modern processors support some form of SIMD parallelism with vectors of a small number of fields (e.g. Intel's SSE). Depending on the architecture, some of the available operations include inter-field instructions such as \emph{shuffle} (which permutes fields in a vector), \emph{pack} and \emph{unpack} (equivalent to our \comp~and \sprd~operations), inter-field shifts, or global sum (which sums all fields in the vector). The power of multimedia processors was studied by Thorup~\cite{Thorup03}, who modeled these processors as vectors of $k$ fields of $\ell$ bits each. Thorup showed that standard global operations on $(k\times \ell)$-bit words can be implemented using vector instructions and inter-field operations in constant time, and argued that this enables the implementation of fundamental combinatorial algorithms such as sorting, hashing, and algorithms for minimum spanning trees on $(k\times\ell)$-bit integers. 

In contrast to Thorup's work, our main interest is in using the ultra wide word to deal with inputs of regular $w$-bit data objects and to speed up algorithms by being able to operate on more of these objects simultaneously. Moreover, we assume that the wide-word \alu supports the standard operations on the full word from the outset, with no need to simulate them using vector operations. Finally, we explore the consequences of indirect memory addressing at the field level, a feature that is not mentioned in Thorup's model.

The \uwram~model can also be related to Multiple-Instruction-Multiple-Data (MIMD) models, and in particular to the PRAM. Although the \uwram~\alu can only execute one instruction on the wide word, it is conceivable to devise a simulation of a PRAM algorithm on the \uwram. Each block of the wide word in the \uwram~acts like a PRAM processor. Since the \uwram~can only execute one type of instruction at a time, each parallel step of the PRAM algorithm is executed in $\lceil s/w \rceil$ steps on the \uwram, where $s$ is the number of different instructions involved in the PRAM algorithm. For a constant number of different PRAM instructions and a non-constant number of \uwram~blocks $w$, this simulation results in a constant overhead in time (compared to the PRAM algorithm running on $\Theta(w)$ processors). However, if such simulation were to be done in any practical implementation of these two models, the actual slowdown would be significant and most instructions would execute serially (as the number of different PRAM instructions is in the same order of magnitude as $w$). On the other hand, any \uwram~algorithm that runs in time $t+q$, where $q$ is the number of \comp~operations and $t$ is the number of steps involved in the rest of the operations, can be simulated in time $O(t+q\log w)$ on a PRAM with $w$ processors, as $\log w$ steps are necessary to simulate a \comp~operation. 

Although simulations between the \uwram~and other models exist, the idea of introducing the \uwram~is to achieve larger speedups with \wordRAM~algorithms, keeping the programming techniques of this model. In practice, the implementations of PRAM algorithms are usually on asynchronous multi-cores, in which programmers must deal with concurrency issues. The advantage of our model is that we can avoid these issues while obtaining similar speedups to those of multi-cores.
\end{FULL}

\section{Simulation of \fsrambig}
\label{sec:rambo}


In the standard \ram~model of computation memory is organized in registers or words, each
word containing a set of bits. Any bit in a word belongs to that
word only. In contrast, in the \fsram~model~\cite{fredman89}---also known as Random Access Machine with Byte Overlap ({\sc rambo})--- words can overlap, that is, a
single bit of memory can belong to several words. 
The topology of the
memory, i.e., a specification of which bits are contained in which
words, defines a particular variant of the \fsram~model.  
Variants of
this model have been used to sidestep lower bounds for important data
structure problems~\cite{brodnik05,brodnik06b}.

\begin{figure}[!t]
\begin{center}
\begin{FULL}
\includegraphics[scale=0.40]{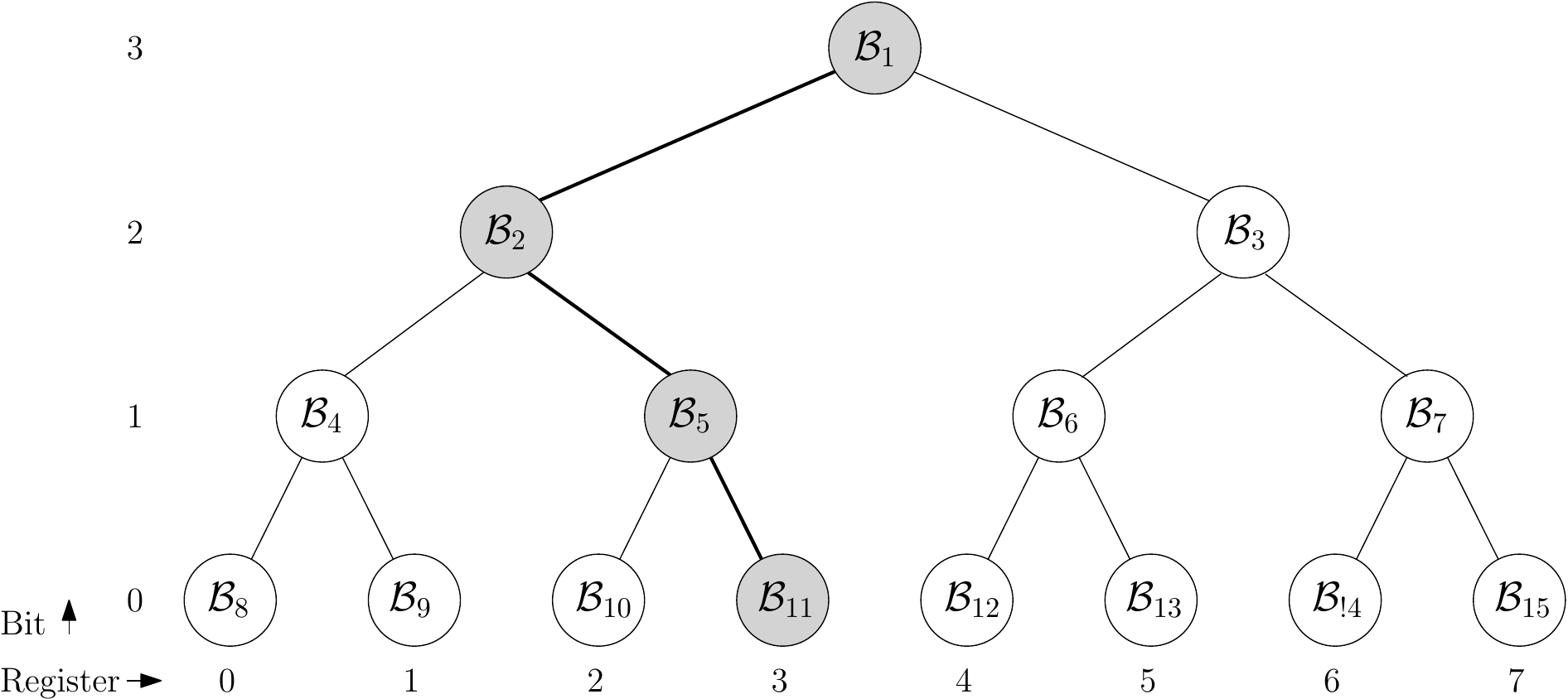}
\end{FULL}
\begin{SHORT}
\includegraphics[scale=0.40]{figs/rambo_tree.pdf}
\end{SHORT}
\caption[Yggdrasil \fsram~memory layout]{Yggdrasil memory layout~\cite{brodnik05}: each node in a complete binary tree
  is an \fsram~bit and registers are defined as paths from a leaf to the
  root. For example, register 3 contains bits $\B_{11},\B_5,\B_2$, and
  $\B_1$ (shaded nodes). 
	}
\label{fig:rambo_tree}
\end{center}
\end{figure}

We show how the \uwram~can be used to implement memory access
operations for any given \fsram~of word size at most $w$ bits in constant
time. Thus, the time bounds of any algorithm in the \fsram~model carry over
directly to the \uwram. Note that each \fsram~layout requires a different specialized hardware implementation, whereas a \uwram~architecture can simulate any \fsram~layout without further changes to its memory architecture.

\begin{FULL}
\subsection{Implementing \fsrambig~Operations in the UW-RAM}
\end{FULL}

Let $\B_1,\ldots,\B_B$ denote the bits of \fsram~memory. A particular \fsram~memory layout can be defined by 
the registers and the bits contained in them~\cite{brodnik95}. For example, in the \emph{Yggdrasil} model
in Fig.~\ref{fig:rambo_tree},
reg[0]=$\B_8\B_4\B_2\B_1$, and in general reg[$i$].bit[$j$]$=\B_k$,
where $k=\lfloor i/2^j\rfloor+2^{m-j-1}$ ($m=4$ in the example)~\cite{brodnik05}.\looseness=-1

In order to implement memory access operations on a given \fsram~using
the \uwram, we need to represent the memory layout of \fsram~in
standard \ram. Assume an~\fsram~memory of $r$ registers of $b\le w$ bits each and $B\le
br$ distinct \fsram~bits. 
We assume that the \fsram~layout is given as a
table $\R$ that stores, for each register and bit within the register,
the number of the corresponding \fsram~bit. Thus, if
reg[$i$].bit[$j$]$=\B_k$, for some $k$, then $\R[i,j]=k$. 
We assume  $\R$ is stored in
row major order.
%
We simply store the value of each \fsram~bit $\B_i$ in a different $w$-bit entry of an array $A$ in~\ram, i.e., $A[i]=\B_i$. 
\begin{FULL}
We could store more than one bit in each word of $A$;
however, this representation allows us to avoid having to serialize concurrent writes to the
same word.
\end{FULL}

\begin{FULL}
Given an index $t$ of a register of an \fsram~represented by $\R$, we
can read the values of each bit of reg[$t$] from \ram~and return the $b$
bits in a word. Doing this sequentially for each bit might take $O(b)$
time. Using the wide word we can take advantage of parallel reading
and the \comp~operation to retrieve the contents of reg[$t$] in
constant time. 
\end{FULL}
\begin{SHORT}
Given an index $t$ of a register of an \fsram~represented by $\R$, we
can read the values of each bit of reg[$t$] from \ram~and return the $b$
bits in a word in constant time using the parallel reading
and \comp~operations.
\end{SHORT}
Let reg[$t$]$=\B_{i_0}\ldots\B_{i_{b-1}}$. The read
operation first obtains the address in $A$ of each bit of register $t$ from $\R$. Then, it uses a content access to read the value of each bit $\B_{i_j}$ into block $W_j$ of
$W$, thus assigning $W_j\gets A[\R[t,j]]$. Finally, it applies one
\comp~operation, after which the $b$ bits are stored in
$W_0$. \begin{FULL} Algorithm~\ref{alg:rambo_read}  
shows the read operation, which takes constant time.\end{FULL}
In order to implement the write operation
reg[$t$]$\gets \B_{i_0}\ldots\B_{i_{b-1}}$ of \fsram, we first set
$W_0\gets \B_{i_0}\ldots\B_{i_{b-1}}$ and perform a \sprd~operation 
 to place each bit $\B_j$ in block $W_j$. We then write the
contents of each $W_j$ in
$A[\R[t,j]]$. \begin{FULL}Algorithm~\ref{alg:rambo_write} 
shows this operation, which takes constant
time as well.\end{FULL}
\begin{SHORT} Both read and write take constant time. We describe these operations in pseudocode in the full version~\cite{uwram_full}.\end{SHORT}

\begin{FULL}
\begin{algorithm}[h]
\caption[\fsram~read]{\fsrasmall\_read($t$)}
\begin{algorithmic}[1]
				\STATE read\_word$(W,\R[t])$ \label{line:read1} \COMMENT{$W_j\gets \R[t,j]$}
				\STATE read\_content$(W,A)$ \label{line:read2} \COMMENT{$W_j\gets A[\R[t,j]]$} \label{lin:rambo_read_2}
				\STATE W $\gets$ \comp$(W)$
				\STATE write\_block$(W,0,\&ret)$ \COMMENT{$ret\gets W_0$}
				\RETURN $ret$
\end{algorithmic}
\label{alg:rambo_read}
\end{algorithm}


\begin{algorithm}[h]
\caption[\fsram~write]{\fsrasmall\_write($t,\B=\B_{i_0}\ldots\B_{i_{b-1}}$)}
\begin{algorithmic}[1]
				\STATE read\_block$(W,0,\B)$ \COMMENT{$W_0\gets \B$}
				\STATE W $\gets$ \sprd$(W)$
				\STATE read\_word$(V,\R[t])$  \COMMENT{$V_j\gets \R[t,j]$}
				\STATE write\_content$(W,V,A)$  \COMMENT{$A[\R[t,j]]\gets W_j$}
\end{algorithmic}
\label{alg:rambo_write}
\end{algorithm}
\end{FULL}

Since the read and write operations described above are sufficient to
implement any operation that uses \fsram~memory (any other operation is
implemented in \ram), we have the following result\begin{SHORT} (see~\cite{uwram_full} for the proof)\end{SHORT}.

\begin{restatable}[]{theorem}{fsramthm}
Let $\R$ be any \fsram~memory layout of $r$ registers of at most $b$
bits each and $B$ distinct \fsram~bits, with $b\le w$ and $\log
B\le w$. Let $A$ be any \fsram~algorithm that uses $\R$ and runs in
time $T$.  Algorithm $A$ can be implemented in the \uwram~to run in
time $O(T)$, using $r b+B$ additional words of \ram.\looseness=-1
\label{thm:rambo}
\end{restatable}

\begin{FULL}
\begin{proof}
Table $\R$ indicating the \fsram~bit identifier for each register and
bit within register can be stored in $rb$ words of \ram, while the
values of each bit can be stored in $B$ words of \ram. Since both
\fsrasmall\_read and \fsrasmall\_write are constant time operations, any
$t$-time operation that uses \fsram~memory can be implemented in \uwram~in the same time $t$. 
\qed
\end{proof}
\end{FULL}

\begin{SHORT}
\paragraph{\textbf{Constant Time Priority Queue}}
\citet{brodnik05} use the Yggdrasil \fsram~memory
\end{SHORT}
\begin{FULL}
\subsection{Constant Time Priority Queue}
\citet{brodnik05} use the Yggdrasil \fsram~memory
\end{FULL}
layout to implement priority queue operations in constant time using
$3M-1$ bits of space ($2M$ of ordinary memory and $M-1$ of \fsram~memory),
where $M$ is the size of the universe. This problem has non-constant
lower bounds for several models\begin{SHORT}, including the \ram~model~\cite{Beame02optimalbounds}.\end{SHORT}
\begin{FULL}, including an
		$\Omega(\min\{\lg \lg M/\lg\lg\lg M,\sqrt{\lg
    N/\lg \lg N}\})$ 
		lower bound in the \ram~model when
the memory is restricted to $N^{O(1)}$, where $N$ is the number of
elements in the set to be maintained~\cite{Beame02optimalbounds}.\end{FULL}
For a universe of size $M=2^m$, for some $m$, the Yggdrasil \fsram~layout consists of
$r= M/2$ registers of $b=\log M$ bits each, and $B=M-1$ distinct
\fsram~bits (Fig.~\ref{fig:rambo_tree} is an example with
$M=16$). Thus, applying Theorem~\ref{thm:rambo} we obtain the
following result:


\begin{corollary}
The discrete extended priority queue problem can be solved in the \uwram~in $O(1)$
time per operation using $2M+w(M/2)\log M+w(M-1)$ bits,
thus in $O(M\log M)$ words of	 \ram.
\end{corollary}

\begin{SHORT}
\paragraph{\textbf{Constant Time Dynamic Prefix Sums}}
\label{sec:dynamic-prefix-sums}
\citet{brodnik06b} use a modified version of the
\end{SHORT}
\begin{FULL}
\subsection{Constant Time Dynamic Prefix Sums}
\label{sec:dynamic-prefix-sums}
\citet{brodnik06b} use a modified version of the
\end{FULL}
Yggdrasil \fsram~to solve the dynamic prefix sums problem in constant
time. This problem consists of maintaining an array
$A$ of size $N$ over a universe of size $M$ that supports the operations $update(j,d)$, which sets
$A[j]$ to $A[j]\oplus d$, and $retrieve(j)$, which returns
$\oplus_{i=0}^jA[i]$~\cite{fredman82,brodnik06b}, where $\oplus$ is any
associative binary operation. This \fsram~implementation sidesteps lower
bounds on various \begin{SHORT}models~\cite{fredman82,hampapuramF98}.\end{SHORT}
\begin{FULL}models: there is an
$\Omega(\log N)$ algebraic complexity lower bound~\cite{fredman82} as
well as under the semi-group model of
computation~\cite{hampapuramF98}, and an $\Omega(\log N/\log \log N)$
information-theoretic lower bound~\cite{fredman82}.\end{FULL}  \begin{SHORT} See the full version~\cite{uwram_full} for more details.

\begin{corollary}
The  operations of the dynamic prefix sums problem
can be supported in $O(1)$ time in the \uwram~with
$O(M^{\sqrt{\log N}})$ bits of \ram.\looseness=-1
\end{corollary}
\end{SHORT}
 
\begin{FULL}
The result of \citet{brodnik06b} uses a complete binary
tree on top of array $A$ as leaves. The tree is similar to the one
used in the priority queue problem, but it differs in that only
internal nodes store any information and in that there are $m=\lceil\log
M\rceil$ bits stored in each node. This tree is stored in a variant of the Yggdrasil memory
called $m$-Yggdrasil, in which each register corresponds again to a
path from a leaf to the root, but this time each node stores not only
one bit but the $m$ bits containing the sum of all values in the leaves of the left
subtree of that node~\cite{brodnik06b}. It is assumed that $nm\le w$,
where $n=\lceil\log N\rceil$ and $w$ is the size of the word in
bits. Thus, an entire path from leaf to root fits in a word and can be
accessed in constant time. An update or retrieve operation consists of
retrieving the values along a path in the tree and processing them in
constant time using bit-parallelism and table lookup operations. The
space used by the lookup table can be reduced at the expense of an
increased time for the retrieve operation. In general, both operations
can be supported in time $O(\iota+1)$ with $(N-1)m$ bits of
$m$-Yggdrasil memory and $O(M^{n/2^\iota}\cdot m+m)$ bits of
\ram, where $\iota$ is a trade-off parameter~\cite{brodnik06b}.

In order to represent the $m$-Yggdrasil memory in our model, we treat
each bit of a node in the tree as a separate \fsram~bit. Thus, the \fsram~memory has $r=N$ registers of $b=nm$ bits each, and there are
$B=(N-1)m$ distinct bits to be stored. Hence, by Theorem~\ref{thm:rambo}
we have:

\begin{corollary}
The operations update and retrieve of the dynamic prefix sums problem
can be supported in the \uwram~model in $O(\iota+1)$ time with $O(M^{n/2^\iota}\cdot
m+Nmnw)$ bits of \ram. For constant time operations ($\iota=1$) the
space is dominated by the first term, i.e., the space is
$O(M^{\sqrt{\log N}})$ bits. For $\iota=\log\log N$, the time is
$O(\log \log N)$ and the space is $O(Nmnw)$ bits.
\end{corollary}

\end{FULL}

\section{Dynamic Programming}
\label{sec:dp}


In this section we describe \uwram~implementations of dynamic programming algorithms for the subset sum, knapsack, and longest common subsequence problems. 
A \wordRAM~algorithm that only uses bit parallelism can be translated directly to the \uwram. The algorithm for subset sum is an example of this. 
In general, however, \wordRAM~algorithms that use lookup tables cannot be directly extended to $w^2$ bits, as this would require a mechanism to address $\Theta(w^2)$-bit words in memory as well as lookup tables of prohibitively large size. Hence, extra work is required to simulate table lookup operations. The knapsack implementation that we present is a good example of such case. 
\begin{SHORT}
We note that these problems have many generalizations that
can be solved using the same techniques and describe them further in the full version~\cite{uwram_full}. 
\end{SHORT}

\begin{SHORT}
\paragraph{\textbf{Subset Sum}}
Given a set $S=\{a_1,a_2,\ldots,a_n\}$ of nonnegative integers
\end{SHORT}
\begin{FULL}
\subsection{Subset Sum}
Given a set $S=\{a_1,a_2,\ldots,a_n\}$ of nonnegative integers
\end{FULL}
(weights) and an integer $t$ (capacity), the subset sum problem is to
find $S'\subseteq S$ such that $\sum_{a_i\in S'}a_i=t$\begin{SHORT}~\cite{CLRS01}.\end{SHORT}
\begin{FULL} The
optimization version asks for the solution of maximum weight which
does not exceed $t$~\cite{CLRS01}.\end{FULL}
This problem is $\NP$-hard, but it can
solved in pseudopolynomial time via dynamic programming in $O(nt)$
time, using the following recurrence~\cite{bellman57}: for
each $0 \le i \le n$ and $0\le j \le t$, $C_{i,j}=1$ if and only if there is
a subset of elements $\{a_1,\ldots,a_i\}$ that adds up to $j$. Thus, 
$C_{0,0}=1$, $C_{0,j}=0$ for all $j>0$, and $C_{i,j}=1$ if
$C_{i-1,j}=1$ or $C_{i-1,j-a_i}=1$ ($C_{i,j}=0$ for
any $j<0$). The problem admits a solution if $C_{n,t}=1$.

\citet{pisinger03} gives an algorithm that implements this recursion in the
word-\ram~with word size $w$ by representing up to $w$ entries of
a row of $C$. Using bit parallelism, $w$ bits of a row can be updated
simultaneously in constant time from the entries of the previous
row: $C_{i}$ is updated by computing
$C_i=(C_{i-1} \ | \ (C_{i-1} >> a_i))$ (which might require shifting
words containing $C_{i-1}$ first by $\lfloor a_i/w \rfloor$ words and
then by $a_i-\lfloor a_i/w \rfloor$)~\cite{pisinger03}. Assuming $w=\Theta(\log t)$,
this approach leads to an $O(nt/\log t)$ time solution in $O(t/\log t)$
space. 
\begin{FULL}The actual elements in $S'$ that form the solution can be 
recovered with the same space and time bounds with a recursive
technique by~\citet{pferschy99}.\end{FULL}

This algorithm can be implemented directly in the \uwram:
entries of row $C_i$ are stored contiguously in memory; thus, we can
load and operate on $w^2$ bits in $O(1)$ time when updating each
row. Hence, the \uwram~implementation runs in $O(nt/\log^2 t)$ time
using the same $O(t/\log t)$ space (number of $w$-bit words).

\begin{SHORT}
\paragraph{\textbf{Knapsack}}
Given a set $S$ of $n$ elements with weights and values, the knapsack
\end{SHORT}
\begin{FULL}
\subsection{Knapsack}
Given a set $S$ of $n$ elements with weights and values, the knapsack
\end{FULL}
problem asks for a subset of $S$ of maximum value such that the total
weight is below a given capacity bound $b$. Let
$S=\{(w_i,v_i)\}_{i=1}^{n}$, where $w_i$ and $v_i$ are the weight and
value of the $i$-th element. 
Like subset sum, this problem is $\NP$-hard but can be solved in pseudopolynomial time using
the following recurrence~\cite{bellman57}: let $C_{i,j}$ be
the maximum value of a solution containing elements in the subset
$S_i=\{(w_k,v_k)\}_{k=1}^{i}$ with maximum capacity $j$. Then,
$C_{0,j}=0$ for all $0\le j\le b$, and
$C_{i,j}=\max\{C_{i-1,j},C_{i-1,j-w_i}+v_i\}$. The value of the
optimal solution is $C_{n,b}$. This leads to a dynamic program that runs in $O(nb)$ time.\looseness=-1 

The word-\ram~algorithm by~\citet{pisinger03} represents partial solutions of the
dynamic programming table with two binary tables $g$ and $h$ and
operates on $O(w)$ entries at a time. More
specifically, $g_{i,u}=1$ and $h_{i,v}=1$ if and only if
there is a solution with weight $u$ and value $v$ that is not
dominated by another solution in $C_{i,*}$ (i.e., there is no entry
$C_{i,u'}$ such that $u'< u$ and $C_{i,u'}\ge v$).  Pisinger shows how
to update each entry of $g$ and $h$ with a constant time procedure,
which can be encoded as a constant size lookup table $T$. A new lookup table $T^\alpha$ is obtained as the product of $\alpha$ times the original table $T$. Thus, $\alpha$ entries of $g$ and $h$ can be computed in constant time. Setting $\alpha = w/10$, an entire row of $g$ and $h$ can be computed in
$O(m/w)$ time and $O(m/w)$ space~\cite{pisinger03}, where $m$ is the maximum of the
capacity $b$ and the value of the optimal solution\begin{FULL}\footnote{This value
  is not known in advance, though an upper bound of at most twice the
  optimal value can be used~\cite{pisinger03,dantzig57}.}\end{FULL}. 
	The optimal
solution can then be computed in $O(nm/w)$ time.

Compared to the subset sum algorithm, which relies mainly on
bit-parallel operations, this word-\ram~algorithm for knapsack relies
on precomputation and use of lookup tables to achieve a $w$-fold speedup. 
While we
cannot precompute a composition of $\Theta(w^2)$ lookup tables to
compute $\Theta(w^2)$ entries of $g$ and $h$ at a time, we can use the
same tables with $\alpha=w/10$ as in Pisinger's algorithm and use the \emph{read\_content} operation
of the \uwram~to make $w$ simultaneous lookups to the
table. Since the entries in a row $i$ of $h$ and $g$ depend only on entries in row
$i-1$, then there are no dependencies between entries
in the same row.

One difficulty is that in order to compute the entries in row $i$ in
parallel we must first preprocess row $i-1$ in both $h$ and $g$, such
that we can return the number of one bits in both
$g_{i-1,0},...,g_{i-1,j}$ and $h_{i-1,0},...,h_{i-1,j}$ in $O(1)$ time
for any column $j \in \{0,m-1\}$.  That is, the prefix sums of the one
bits in row $i-1$.  Note that 
this is
\emph{not} the same as the dynamic problem described in
Section~\ref{sec:dynamic-prefix-sums}, but it is a static prefix sums
problem. 
\begin{SHORT}
We describe how to compute the prefix sums of a row of $g$ and $h$ in $O(m/w^2)$ time in the full version~\cite{uwram_full}.\looseness=-1
\end{SHORT}
\begin{FULL}
Furthermore, since the algorithm is the same for both $g$ and
$h$, we describe the computation for $g$ alone.\looseness=-1

\paragraph{Static Prefix Sums} 
We divide $g_{i-1}$ in blocks of $w$ contiguous bits and compute the number of ones in each block 
$g_{i-1,k},...,g_{i-1,k+w-1}$ for $k \in
\{0,w,2w,...,\lfloor m/w\rfloor w \}$ using a lookup table. We store
the results in an array $\prefixsumarray$ of length $\lceil m/w
\rceil$, with $A[k]$ storing the number of ones in the $k$-th block.  Next, we compute the prefix sums $\prefixsumarray'$ of $\prefixsumarray$ in
two steps.  We divide $\prefixsumarray$ in subarrays of $w$ consecutive entries. Let $\prefixsumarray_i$ denote the subarray $\prefixsumarray[iw,iw+w-1]$, for $i\in \{0,1,\ldots,\lceil|\prefixsumarray|/w\rceil-1\}$.

The first step is to compute the prefix sums $\prefixsumarray'_i$ of each subarray $\prefixsumarray_i$, i.e. $\prefixsumarray_i'[k]=\sum_{j=0}^k \prefixsumarray_i[j]$.
Using the $w$ blocks of a wide word, we can operate on $w$ entries at a time. Consider the first $w$ consecutive subarrays $\prefixsumarray_{0},\prefixsumarray_1,\ldots,\prefixsumarray_{w-1}$. In order to compute $\prefixsumarray'_0,\ldots,\prefixsumarray'_{w-1}$, for each $0\le k \le w-1$, we use the $i$-th block of the wide work to compute $\prefixsumarray'_{i}[k]$, thus computing the entries for all $0\le i\le w-1$ simultaneously. Each entry is computed in constant time, since  
\begin{equation} \notag
\prefixsumarray_i'[k] = \begin{cases} \prefixsumarray_i'[k-1] + \prefixsumarray_i[k] & \text{if $k>0$,}\\\prefixsumarray_i[k] & \text{otherwise.} \end{cases}
\end{equation} 
Hence, we can compute the prefix sums of $w$ subarrays in $O(w)$ time. After computing the first $w$ subarrays we continue with the second group, and so on. Thus, we compute all prefix sums of the $O(|\prefixsumarray|/w)$ subarrays in $O(|\prefixsumarray|/w)$ time.

The second step is to update each subarray of $\prefixsumarray'$ by adding to each entry the last entry of the previous subarray. I.e., we set $\prefixsumarray'_i[k] = \prefixsumarray'_{i}[k] + \prefixsumarray'_{i-1}[w-1]$ for all $i=1,\ldots,\lceil|\prefixsumarray'|/w\rceil -1$ (in increasing value of $i$).  This can also be done for $w$ entries at once, but this time we use the blocks of the wide word to update all entries of one subarray simultaneously. Thus, sequentially for each $i=1,\ldots,\lceil|\prefixsumarray'|/w\rceil -1$  we update $\prefixsumarray'_i$ in $O(1)$ time, and hence $\prefixsumarray'$ is updated in $O(|\prefixsumarray|/w)$ time.

At this point, $\prefixsumarray'$
contains the prefix sums of $\prefixsumarray$, and took
$O(|\prefixsumarray|/w) = O(m/w^2)$ time to compute. 
Fig.~\ref{fig:prefix_sums} shows an example of this procedure.

Let $f$ be the number of ones in $g_{i-1,\lfloor j/w\rfloor}, ...,
g_{i-1,j}$, which can be computed using the lookup table.  To compute the number of ones in 
$g_{i-1,0},...,g_{i-1,j}$ we return $f + \prefixsumarray'[\lfloor j/
  w\rfloor]$.  

\begin{figure}[!t]
\begin{center}
\includegraphics[scale=0.80]{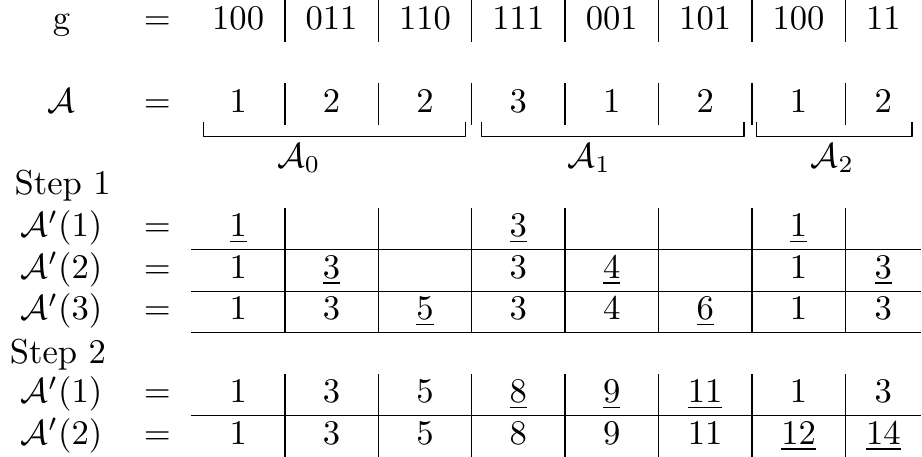}
\caption[Prefix sums example]{Example of computing prefix sums in the \uwram~with $w=3$ and $m=23$. Numbers in parenthesis indicate the parallel step number when computing $\prefixsumarray'$ and underlined entries indicate the entries computed in that step. 
	}
\label{fig:prefix_sums}
\end{center}
\end{figure}

\end{FULL}
Then, each row of $g$ and $h$ takes $O(m/w^2)$ time to
compute, and since there are $n$ rows, the total time to compute $g$ and $h$ (and hence the optimal solution)
on the \uwram~is $O(nm/w^2)$.  This achieves a $w$-fold speedup over
Pisinger's word-\ram~solution.

\begin{FULL}
\subsection{Generalizations of Subset Sum and Knapsack Problems}
\label{sec:knap-gen}

\citet{pisinger03} uses the techniques of the word-\ram~
algorithm for subset sum and knapsack to obtain a word-\ram~algorithm
for computing a path in a layered network: given a graph $G=(V,E)$, a
source $s\in V$ and a terminal $t\in V$, and a weight for each edge,
is there a path of weight $b$ from $s$ to $t$? Again, this algorithm
translates directly to a \uwram~algorithm, thus yielding a $w$-fold speedup
over the word-\ram~algorithm. Pisinger further uses the algorithms for the  problems above to 
implement word-\ram~solutions for other generalizations of subset sum
and knapsack problems, such as: the bounded subset sum and knapsack
problems (each element can be chosen a bounded number of times), the
multiple choice subset sum and knapsack problems (the set of numbers
is divided in classes and the target sum must be matched with one
number of each class), the unbounded subset sum and knapsack problems
(each element can be chosen an arbitrary number of times), the
change-making problem, and, finally, the two-partition problem. \uwram~implementations for all these generalizations are direct and yield a
$w$-fold speedup over the word-\ram~algorithms (recall that $w=\Omega(\log n)$).
\end{FULL}

\begin{SHORT}
\paragraph{\textbf{Longest Common Subsequence}}
\label{sec:uwram_lcs}
The final dynamic programming problem we examine is that of computing
\end{SHORT}
\begin{FULL}
\subsection{Longest Common Subsequence}
The final dynamic programming problem we examine is that of computing
\label{sec:uwram_lcs}
\end{FULL}
the longest common subsequence (LCS) of two string sequences 
\begin{SHORT}
(see the full version~\cite{uwram_full} for a definition).
\end{SHORT}
\begin{FULL}
(Definition~\ref{def:lcs}).
\begin{definition}
\label{def:lcs}[LCS]
Given a
sequence of symbols $X=x_1x_2\ldots x_m$, a sequence $Z=z_1z_2\ldots
z_k$ is a subsequence of $X$ if there exists an increasing sequence of
indices $i_1,i_2,\ldots,i_k$ such that for all $1\le j \le k$,
$x_{i_j}=z_j$~\cite{CLRS01}. 
Let $\Sigma$ be a finite alphabet of
symbols, and let $\sigma = |\Sigma|$. Given two sequences
$X=x_1x_2\ldots x_m$ and $Y=y_1y_2\ldots y_n$, where $x_i,y_j\in
\Sigma$, the Longest Common Subsequence problem asks for a sequence $Z=z_1z_2\ldots z_k$ of
maximum length such that $Z$ is a subsequence of both $X$ and $Y$.\looseness=-1
\end{definition}

\end{FULL}
This problem can be solved via a classic dynamic programming algorithm
in $O(nm)$ time~\cite{CLRS01}. 
\begin{SHORT}
In~\cite{uwram_full} we \end{SHORT}
\begin{FULL}We \end{FULL} describe a \uwram~algorithm for LCS based on an algorithm by Masek and Paterson~\cite{masek80}. We note that there exist other approaches to solving the LCS problem with bit-parallelism (e.g.,~\cite{Crochemore01}) that could also be adapted to work in the \uwram. The approach we show here is a good example of bit parallelism combined with the parallel lookup power of the model, which we use to implement the Four Russians technique. 
\begin{SHORT} 
We obtain the following results:\looseness=-1
\end{SHORT}

\begin{FULL} The base algorithm, which mainly relies on bit parallelism, leads to Theorem~\ref{thm:lcs1}. We then extend the algorithm with the Four Russians technique to achieve further speedups, obtaining Theorem~\ref{thm:lcs2}.\end{FULL}

\begin{restatable}[]{theorem}{lcsa}
\label{thm:lcs1}
The length of the LCS of two strings  $X$
and $Y$ over an alphabet of size $\sigma$, with $|X|=m$ and $|Y|=n$,
can be computed in the \uwram~in
$O(\frac{nm}{w^2}\log \sigma+m+n)$ time and $O(\frac{\min(n,m)}{w}\log \sigma)$
words in addition to the input.
\end{restatable}

\begin{restatable}[]{theorem}{lcsb}
\label{thm:lcs2}
The length of the LCS of two strings  $X$
and $Y$ of length $n$ over an alphabet of size $\sigma$
can be computed in the \uwram~in
$O(n^2\log^2(\sigma)/w^3+n\log(\sigma)/w)$ time. For $\sigma=O(1)$ and
$w=\Theta(\log n)$ this time is $O(n^2/\log^3 n)$.
\label{thm:four_russians_uw_ram}
\end{restatable}

\begin{FULL}

Let $c_{i,j}$ denote the length of the
LCS of $X[1..i]=x_1x_2\ldots x_i$ and $Y[1..j]=y_1y_2\ldots y_j$. Then
the following recurrence allows us to compute the length of the LCS of
$X$ and $Y$~\cite{CLRS01}:

\begin{equation}
c_{i,j}=\left\{\begin{array}{ll}
	0, &  \textrm{ if } i=0 \textrm{ or } j=0\\
	c_{i-1,j-1}+1, & \textrm{ if } x_i=y_j \\
	\max\{c_{i,j-1},c_{i-1,j}\}, & \textrm{ otherwise. }
	   \end{array}\right.
	   \label{eqn:lcs_rec}
\end{equation}

The length of the LCS is $c_{m,n}$, which can be computed in $O(mn)$
time. Consider an $(m+1)\times (n+1)$ table $C$ storing the values
$c_{i,j}$. The idea of the \uwram~algorithm is to compute various
entries of this table in parallel. We assume $w=\Theta(\max\{\log n,
\log m\})$.

Let $d_k$ denote the values in the $k$-th diagonal of table $C$,
this is $d_k=\{c_{i,j}| i+j=k\}$. Since a value in a cell $i,j>0$
depends only on the values of cells $(i-1,j)$, $(i-1,j-1)$ and
$(i,j-1)$, all values in the same diagonal $d_k$ can be computed in parallel.  Thus, we use the wide
word to compute various entries of a diagonal in constant time. Since
each value in the cell might use up to $\min\{\log n,\log m\}$ bits,
each value might use up to an entire block of the wide word (if
$\log m=\Theta(\log n)$); thus, $w$ cells can be computed in parallel. Since the
total number of cells is $O(mn)$ and the critical path of the table has
$m+n+1$ cells, this approach takes $O(mn/w+m+n)$ parallel time,
resulting in a speedup of $w$. However, we can obtain better speedups
by using fewer bits per entry of the table, which enables us to
operate on more values in parallel. For this sake, instead of storing
the actual values of the partial longest common subsequences, we store
differences between consecutive values as described in~\cite{masek80}
for the related string edit distance problem.

\begin{figure*}[t]
\begin{center}
\small
\begin{tabular}{p{4.2cm} p{0.05cm} p{4.2cm} p{0.05cm} p{4.2cm}}
\begin{tabular}{p{0.3cm} p{0.3cm} | p{0.2cm}  p{0.2cm} p{0.2cm} p{0.2cm} p{0.2cm} p{0.2cm} p{0.2cm} }
& &	j &	1&	2&	3&	4& 5&	6\\
\multicolumn{2}{c|}{LCS} &	&  a&	a&	b&	b&	b&	a\\\hline
i & &	0&	0&	0&	0&	0&	0&	0\\
1 & a &	0	&1	&1	&1&	1&	1&	1\\
2 & b&	0&	1&	1&	2&	2&	2&	2\\
3 & b&	0&	1&	1&	2&	3&	3&	3\\
4& a&	0&	1&	2&	2&	3&	3&	4\\
5 &b&	0&	1&	2&	3&	3&	4&	4
\end{tabular}
&

&
\begin{tabular}{p{0.3cm} p{0.3cm} | p{0.2cm}  p{0.2cm} p{0.2cm} p{0.2cm} p{0.2cm} p{0.2cm} p{0.2cm} }
& &	j &	1&	2&	3&	4& 5&	6\\
\multicolumn{2}{c|}{H} &	&  a&	a&	b&	b&	b&	a\\\hline
i &	&&	0&	0&	0&	0&\cellcolor[gray]{0.9}	0&\cellcolor[gray]{0.6}	0\\
1 & a& &		1&	0&	0&\cellcolor[gray]{0.9}	0&\cellcolor[gray]{0.6}	0&	0\\
2 & b	&&	1&	0&\cellcolor[gray]{0.9}	1&\cellcolor[gray]{0.6}	0&	0&	0\\
3& b	&&	1&	\cellcolor[gray]{0.9}0&\cellcolor[gray]{0.6}	1&	1&	0&	0\\
4 & a	&&\cellcolor[gray]{0.9}	1&\cellcolor[gray]{0.6}	1&	0&	1&	0&	1\\
5 & b	&&\cellcolor[gray]{0.6}	1&	1&	1&	0&	1&	0\\
\end{tabular}
&

&
\begin{tabular}{p{0.3cm} p{0.3cm} | p{0.2cm}  p{0.2cm} p{0.2cm} p{0.2cm} p{0.2cm} p{0.2cm} p{0.2cm} }
& &	j &	1&	2&	3&	4& 5&	6\\
\multicolumn{2}{c|}{V} &	&  a&	a&	b&	b&	b&	a\\\hline
i &	&&	&	&	&	&	&	\\
1 & a& 0&		1&	1&	1&\cellcolor[gray]{0.9}	1&	1&	1\\
2 & b	&0&	0&	0&\cellcolor[gray]{0.9}	1&	1&	1&	1\\
3& b	&0&	0&\cellcolor[gray]{0.9}	0&	0&	1&	1&	1\\
4 & a	&0&\cellcolor[gray]{0.9}	0&	1&	0&	0&	0&	1\\
5 & b	&\cellcolor[gray]{0.9}0&	0&	0&	1&	0&	1&	0\\\end{tabular}
\end{tabular}
\normalsize
\end{center}
\caption[Example of dynamic programming tables for Longest Common Subsequence]{Dynamic programming tables for the LCS and horizontal and vertical differences for $X=abbab$ and $Y=aabbba$.}
\label{fig:lcs}
\end{figure*}

Let $V$ and $H$ denote the tables of vertical and horizontal differences of values in $C$, respectively. Entries in these tables are defined as $V_{i,j}=c_{i,j}-c_{i-1,j}$ and $H_{i,j}=c_{i,j}-c_{i,j-1}$ for
$1\le i\le m$ and $1\le j\le n$. Fig.~\ref{fig:lcs} shows the tables $C$, $V$, and $H$ for an example pair of input sequences. 
We adapt Corollary 1
in~\cite{masek80} for the computation of $V$ and $H$:

\begin{restatable}[]{proposition}{lcsprop}
Let $[x_i=y_j]=1$ if $x_i=y_j$ and 0 otherwise. Then,
$V_{i,j}=\max\{[x_i=y_j]-H_{i-1,j},0,V_{i,j-1}-H_{i-1,j}\}$ and
$H_{i,j}=\max\{[x_i=y_j]-V_{i,j-1},0,\\ H_{i-1,j}-V_{i,j-1}\}$.\looseness=-1
\label{prop:lcs}
\end{restatable}

\begin{proof}
Directly from Recurrence~(\ref{eqn:lcs_rec}) we obtain
$V_{i,j}=1-H_{i-1,j}$ if $x_i=y_j$ and
$V_{i,j}=\max\{0,V_{i,j-1}-H_{i-1,j}\}$ otherwise. Similarly,
$H_{i,j}=1-V_{i,j-1}$ if $x_i=y_j$ and
$H_{i,j}=\max\{0,H_{i-1,j}-V_{i,j-1}\}$ otherwise. It is easy to
verify from the definition of longest common subsequence and
Recurrence~(\ref{eqn:lcs_rec}) that $0\le H_{i,j}\le 1$ and $0\le
V_{i,j}\le 1$ for all $i,j$, which implies that the maximum in
$\max\{[x_i=y_j]-H_{i-1,j},0,V_{i,j-1}-H_{i-1,j}\}$ and
$\max\{[x_i=y_j]-V_{i,j-1},0,H_{i-1,j}-V_{i,j-1}\}$ is equal to the
first term if $x_i=y_j$ and to the second or third terms otherwise.\qed
\end{proof}

We compute tables $H$ and $V$ according to Proposition~\ref{prop:lcs}
diagonal by diagonal using bit parallelism in the wide word. Assume an
alphabet $\Sigma=\{0,1,2,\ldots,\sigma-1\}$ with $\lceil\log
\sigma\rceil\le w-1$. Although all entries in tables $H$ and $V$ are
either 0 or 1, we will use fields of $O(\log \sigma)$ bits to store
these values, since we can only compare at most $w^2/\log \sigma$    
symbols simultaneously in the wide word. We divide the wide word $W$
in $f$-bit fields with $f=\max(\lceil\log \sigma\rceil,2)+1$. Each
field will be used to store both symbols and intermediate results for
the computation of the diagonals of $H$ and $V$, plus an additional bit to
serve as a test bit in order to implement fieldwise comparisons as
described in Appendix~\ref{app:subroutines}. 
We require at least 3 bits
because although all entries in tables $H$ and $V$ use one bit,
intermediate results in calculations can result in values of -1. Thus,
we require 2 bits to represent values -1, 0, and 1, and a test or
sentinel bit to prevent carry bits resulting from subtractions to
interfere with neighboring fields. We represent -1 in two's
complement. It is not hard to extend the techniques for comparisons
and maxima to the case of positive and negative
numbers~\cite{hagerup98}.

Let $H_k$ and $V_k$ denote the $k$-th diagonal of $H$ and $V$,
respectively, i.e., $H_k=\{H_{i,j}| i+j=k\}$ and $V_k=\{V_{i,j}|
i+j=k\}$. Consider table $H$. We will operate with each diagonal $H_k$
using $\lceil |H_k|/\ell \rceil$ wide words, where $\ell=\lfloor w^2/f \rfloor$. Let
$f_0,\ldots,f_{\ell-1}$ denote the fields within a wide word in increasing
order of bit significance. In each wide word, cells of
$H_k$ will be stored in increasing order of column, i.e., if $H_{i,j}$
is stored in field $f_r$, then $f_{r+1}$ stores $H_{i-1,j+1}$. In
order to compute each diagonal we must compare the relevant entries of
strings $X$ and $Y$. We assume that each symbol of $X$ and $Y$ is
stored using $\lceil\log \sigma\rceil+1$ bits (including the test
bit) and that $X$ is stored in reverse order. $X$ and $Y$ can be
preprocessed in $O(m+n)$ to arrange this representation, which will
allow us to do constant-time parallel comparisons of symbols for each diagonal
loading contiguous words of memory in wide words.

Consider a diagonal $H_k$. Assume that the entire diagonal fits in a
word $W$. This will not be the case for most diagonals, but we
describe the former case for simplicity. The latter case is
implemented as a sequence of steps updating portions of the
diagonal that fit in a wide word. We update the entries of $H_k$ as follows:

\begin{enumerate}

\item We load the
symbols of the relevant substrings of $X$ and $Y$ into words $W_X$ and
$W_Y$, with the substring of $X$ in reverse order. More specifically,
for a diagonal $k$, $W_Y=y_{j_1}y_{{j_1}+1}\ldots y_{j_2}$, where
$j_1=k-\min(|X|,k-1)$ and $j_2=\min(|Y|,k)$, 
and $W_X=x_{i_2}x_{i_2-1}\ldots x_{i_1}$ with $i_2=k-j_1$ and
$i_1=k-j_2$. We subtract $W_Y$ from $W_X$, mask out all non-zero
results and write a 1 in each field that resulted in 0. We store the
resulting word in $W_{eq}$, where each field corresponding to a cell
$(i,j)$ stores a 1 if $x_i=y_j$ and a 0 otherwise (this can be implemented
through comparisons as described in 
Appendix~\ref{app:subroutines}). 

\item
We load $V_{k-1}$ into a word $W_V$ and subtract it from $W_{eq}$ to
obtain $[a_i=b_j]-V_{i,j-1}$ for all $i,j$ in $H_k$ simultaneously and
store the result in $W_1$. \item 

We load $H_{k-1}$ into a word $W_H$ and
subtract $W_V$ from it to obtain $H_{i-1,j}-V_{i,j-1}$ for all $i,j$ in
$H_k$, storing the result in $W_2$. 

\item Finally, using fieldwise
comparisons, we obtain the fieldwise maximum of $W_1,W_2$ and the
word $\vec{0}$. The resulting word is $H_k$.

\end{enumerate}

All the operations described above can be implemented in constant time.  
The procedure to compute
$V_k$ is analogous.  Note that the entries corresponding to base cases
in the first row and column in the LCS table correspond to the base
cases of the horizontal and vertical vectors, respectively. When
computing diagonals $H_k$ with $k\le n+1$ and $V_k$ with $k\le m+1$,
the entries corresponding to base cases are not computed from previous
diagonals but should be added appropriately at the end of $H_k$ and
beginning of $V_k$. Example~\ref{ex:lcs}  
 shows how to compute $H_6$ from $H_5$ and $V_5$ (in gray) in Fig.~\ref{fig:lcs} with the above procedure.

\begin{example}
\label{ex:lcs}
Let $X=abbab$ and $Y=aabbba$ be two strings. Fig.~\ref{fig:lcs}
shows the entries of the dynamic programming table for computing the
LCS of $X$ and $Y$, as well as the values of horizontal and vertical
differences.

In this example $\sigma=2$, thus we use one bit for each symbol
(`a'=0, `b'=1), but we use $f=3$ bits per field. Consider the diagonal
$H_6$ in table $H$ (in dark gray). We now illustrate how to obtain $H_6$ from $H_5$ and
$V_5$ (in light gray). In what follows we represent the
number in each field in decimal and do not include the details of
fieldwise comparison and maxima.

\begin{center}
\small
\begin{tabular}{p{2.5cm} p{0.2cm} l  p{5cm}}
$W_X$ & = &  1 \ 0 \ 1 \ 1 \ 0 &(=$x_5x_4x_3x_2x_1$)\\
$W_Y$ & = & 0 \ 0 \ 1 \ 1 \ 1  &(=$y_1y_2y_3y_4y_5$)\\
$W_{eq}$ & =&   0 \ 1 \ 1 \ 1 \ 0  &($W_{eq}[f(j-1)]=1 \Leftrightarrow x_{|H_5|-j}=y_j$)\\
$V_5$ &= &  0 \ 0 \ 0 \ 1 \ 1\\
$W_1=W_{eq}-V_5$& = &  0 \ 0 \ 1 \ 0 -1\\ \hline
$H_5$ & = & 1  \ 0 \  1 \ 0  \ 0&\\
$W_2=H_5-V_5$ & =& 1 \ 0 \ 1 -1 -1\\\hline
$\max\{W_1,W_2,\vec{0}\}$ &=  & 1 \ 1 \ 1 \ 0 \ 0 \\
$H_6$ &= & 1 \ 1 \ 1 \ 0 \ 0\ 0 & (last 0 is the base case)
\end{tabular}
\normalsize
\end{center}
\end{example}

Once all diagonals are computed, the final length of the longest
common subsequence of $X$ and $Y$ can be simply computed by
(sequentially) adding the values of the last row of $H$ or the values
of last column of $V$ (which can be done while computing $H$ and
$V$). The entire procedure is described in Algorithm~\ref{alg:lcs} 
and leads to Theorem~\ref{thm:lcs1}: 

\lcsa*

\begin{proof}
A diagonal of $H$ and $V$ of length $\ell$ entries can be computed in
time $O(\ell\log \sigma/w^2+1)$. Adding this time over all $m+n$
diagonals yields the total time. For the space, each diagonal is
represented in $\lceil \ell f/w^2 \rceil$ wide words, where $f=O(\log
\sigma)$ is the number of bits per field. Since we can compute each
diagonal $H_k$ and $V_k$ using only $H_{k-1}$ and $V_{k-1}$, we only
need to store 4 diagonals at any given time. Since the maximum length
of a diagonal is $\min(n,m)+1$ and each wide word can be stored in $w$
regular words of memory, the result follows.\qed
\end{proof}

\begin{algorithm}[t!]
\caption[UW-RAM LCS-length]{LCS-length($X,Y,m=|X|,n=|Y|,\sigma$)}
\begin{algorithmic}[1]
			\STATE $f\gets\max(\lceil\log \sigma\rceil,2)+1$ \COMMENT{field length in bits}
			\STATE $H_1^1\gets \vec{0}$ \COMMENT{$H_{0,1}=0$}
			\STATE $V_1^1\gets \vec{0}$ \COMMENT{$V_{1,0}=0$}
			\STATE $\mathrm{length} \gets 0$ \COMMENT{length of longest common subsequence}
			\FOR{$k=2$ to $m+n$}
					\STATE $\ell \gets \min(n,k-1)+\min(m,k-1)-k+1$ \COMMENT {length of diagonal} 
					\STATE $j_1 \gets k-\min(m,k-1)$ \COMMENT{indices of relevant substrings of $X$ and $Y$}
					\STATE $j_2 \gets \min(n,k)$
					\STATE $i_2 \gets k-j_1$
					\STATE $i_1 \gets k-j_2$
					\STATE $j\gets j_1$
					\STATE $i\gets i_2$
					\STATE $s \gets \lceil\ell f/w^2\rceil$ \COMMENT {number of wide words per diagonal}
					
					\FOR{$t=1$ to $s$}
						\STATE $j' \gets \min(j+s-1,j_2)$
						\STATE $i' \gets \max(i+s-1,i_1)$
						\STATE $W_Y \gets Y[j..j']$
						\STATE $W_X\gets X[i..i']$ \COMMENT{substring of $X$ is in reverse order}
						\STATE $W_{eq}\gets $equal$(W_X,W_Y)$
						\STATE $W_1\gets W_{eq}-V_{k-1}^t$
						\STATE $W_2\gets H_{k-1}^t-V_{k-1}^t$
						\STATE $H_k^t\gets \max(W_1,W_2,\vec{0})$ \COMMENT{base case is implicitly added at rightmost field}
						\STATE $W_1\gets W_{eq}-H_{k-1}^t$
						\STATE $W_2\gets V_{k-1}^t-H_{k-1}^t$ 
						\STATE $V_k^t\gets \max(W_1,W_2,\vec{0})$ 
						\IF{$t=1$ AND $k\le m+1$}
						\STATE $V_k^t\gets V_k^t >> f$ \COMMENT{add 0 in the first field for the base case}
						\ENDIF
						\STATE $i\gets i'+1$
						\STATE $j\gets j'+1$
						\IF{$t=1$ AND $k\ge m+1$}
								\STATE $\mathrm{length} \gets \mathrm{length}+ H_k^1[0..f-1]$ \COMMENT{$\mathrm{length} = \mathrm{length} + H_{m,k-m}$}
						\ENDIF
					\ENDFOR
			\ENDFOR
			\RETURN $\mathrm{length}$
\end{algorithmic}
\label{alg:lcs}
\end{algorithm}

\subsubsection{Recovering a Longest Common Subsequence}
\label{sec:recovering_lcs}

It is known that given a dynamic programming table storing the values
of the LCS between strings $X$ and $Y$, one can recover the actual
subsequence by starting from $c_{m,n}$ and following the path through the
cells corresponding to the values used when computing each value
$c_{i,j}$ according to Recurrence~(\ref{eqn:lcs_rec}): if $x_i=y_j$,
then we add $x_i$ to the LCS and continue with cell
$(i-1,j-1)$; otherwise the path follows the cell corresponding to the
maximum of $c_{i-1,j}$ or $c_{i,j-1}$. Although
Algorithm~\ref{alg:lcs} does not compute the actual LCS table, a path
of an LCS can be easily computed using tables $H$ and $V$. The path
starts at cell $(m,n)$ (of either table). Then, to continue from a cell
$(i,j)$, if $x_i=y_j$, then $x_i$ is part of the LCS, and we continue
with cell $(i-1,j-1)$; otherwise, if $H_{i,j}=1$ and $V_{i,j}=0$, then
we continue with cell $(i-1,j)$, and if $H_{i,j}=0$ and $V_{i,j}=1$, we
continue with cell $(i,j-1)$ (and with any of the two if
$H_{i,j}=V_{i,j}=0$). This can be easily done in $O(m+n)$ time if all
diagonals of tables $V$ and $H$ are kept in memory while computing the
LCS length in Algorithm~\ref{alg:lcs}. This would require
Algorithm~\ref{alg:lcs} to use $O(nmw/\log \sigma)$ words of memory to
store all diagonals.

\subsubsection{Four Russians Technique}
\label{sec:four_russians_lcs}

The computation of the longest common subsequence in the \uwram~can be made even
faster by combining the diagonal-by-diagonal order of computation
described above with the Four Russians technique.  The Four Russians
technique~\cite{fourRussians70} was used 
by Masek and Paterson to speedup the computation of the string edit problem (and also the LCS) in a \ram~with indirect addressing~\cite{masek80}.
The technique consists of dividing the dynamic programming table in blocks of size
$t\times t$ cells. In a precomputation phase, all possible blocks are
computed and stored as a data structure indexed by the first row and
column of each block. The LCS can be then computed by looking up
relevant values of the table one block at a time using the data
structure. In a \ram~with indirect addressing and under a suitable
value of $t$, the last row and column of a block can be obtained by
looking up the entry corresponding to the first row and column of that
block in constant time. This technique yields a speedup of $O(t^2)$
with respect to computing all cells in the table, for a total time of
$O(n^2/t^2)$ (for two strings of length $n$) plus the time for the
precomputation of all blocks. By setting $t=O(\log n)$ (for a constant alphabet size) and encoding
the table with difference vectors, the precomputation time can be
absorbed by the time to compute the main table
(see~\cite{masek80,gusfield97} for a more detailed description of the
technique).

We can use the power of parallel memory accesses of the \uwram~to
speedup the computation of the LCS even further by looking up blocks
in parallel, in a similar fashion to the diagonal-by-diagonal approach
described above. For simplicity, assume $m=n$. Using the same encoding
for $H$ and $V$, we first precompute all possible blocks of $H$ and
$V$ of size $t\times t$. Since a block is completely determined by its
first column and row, whose values are in $\{0,1\}$, and the two
substrings of length $t$ (over an alphabet of size $\sigma)$, there
are $O((2\sigma)^{2t})$ possible blocks. Note that we can encode each
cell now with one bit, since we do not need to do symbol comparisons
in parallel. Each block can be computed in $O(t^2)$ time with the
standard sequential algorithm, so the precomputation time is
$O((2\sigma)^{2t}t^2)$. We set $t=\log_{2\sigma}n/2$, and thus the
precomputation time is $O(n\log^2n)$~\cite{gusfield97}. Since $t\le
w/2$, we can use each block of the wide word to lookup the entry for
each block by using a parallel lookup operation. Thus, as described
previously, we can compute tables $H$ and $V$ in diagonals of blocks,
computing $\min(\ell,w)$ blocks simultaneously in a diagonal of length
$\ell$ blocks. There are $(n/t)^2$ blocks to compute and the critical
path of the table has length $n/t$ blocks. Therefore, the computation
of $H$ and $V$ can be carried out in time
$O(n^2/(t^2w)+n/t)=O(n^2\log^2\sigma/w^3+n\log\sigma/w)$, since
$t=\Theta(w/\log \sigma)$. This result is summarized by Theorem~\ref{thm:lcs2}:

\lcsb*


\end{FULL}

\section{String Searching}
\label{sec:string_search}
Another example of a problem where a large class of algorithms can be
sped up in the \uwram~is string searching. Given a text $T$ of length
$n$ and a pattern $P$ of length $m$, both over an alphabet $\Sigma$,
string searching consists of reporting all the occurrences
of $P$ in $T$. \begin{FULL}We focus here on on-line
searching, this is, with no preprocessing of the text (though
preprocessing of the pattern is allowed), and we assume in general that
$n\gg m$.\end{FULL}
\begin{SHORT} We assume in general that $n\gg m$.\end{SHORT}
We use two classic algorithms for this problem to illustrate
different ways of obtaining speedups via parallel operations in the
wide word. More specifically, we obtain speedups of $w=\Omega(\log n)$ for \uwram~implementations of the Shift-And and Shift-Or
algorithms~\cite{by92,wu92}, and the Boyer-Moore-Horspool
algorithm~\cite{horspool80}. 
\begin{FULL}For a string $S$, let $S[i]$ denote its
$i$-th character, and let $S[i..j]$ be the substring of $S$
from position $i$ to $j$. 
Indices start at 1. \looseness=-1
\end{FULL}

\begin{SHORT}
\paragraph{\textbf{Shift-And and Shift-Or}}
These algorithms simulate an $(m+1)$-state
non-deterministic automaton that recognizes $P$ starting from every
position of $T$.
\end{SHORT}
\begin{FULL}
\subsection{Shift-And and Shift-Or}
\label{sec:shift_and_or}
The Shift-And and Shift-Or algorithms keep a sliding window of length
$m$ over the text $T$. On a window at substring
$T[i-m+1..i]$, the algorithms keep track of all prefixes of $P$ that
match a suffix of $T[i-m+1..i]$. Thus, if at any time there is one
such prefix of length $|P|$, then an occurrence is reported at
$T[i-m+1]$. This is equivalent to running the $(m+1)$-state
non-deterministic automaton that recognizes $P$ starting from every
position of $T$. 
\end{FULL}
For a window $T[i-m+1..i]$ in $T$, the $j$-th state
of the automaton $(0\le j\le m)$ is active if and only if $P[1..j]=T[i-j+1..i]$. These
algorithms represent the automaton as a bit vector and update the
active states using bit-parallelism.  Their running time is $O(mn/w+n)$, achieving linear time on the size of the text for small patterns.
\begin{FULL} 
More specifically, the Shift-And
algorithm keeps a bit vector $\vec{v}=b_1b_2\ldots b_{m}$, where
$b_j=1$ whenever the $j$-th state is active. If $\vec{v}_i$ represents
the automaton for the window ending at $T[i]$, then
$\vec{v}_{i+1}=((\vec{v}_i>>1)\ | \ 1)\ \& \ Y[T[i+1]]$, where
$Y[\sigma]$ is a bit vector with set bits in the positions of the
occurrences of $\sigma$ in $P$. The OR with a 1 corresponds to the initial state always being active to allow a match to start at any
position. The Shift-Or algorithm is similar but it saves this
operation by representing active states with zeros instead of ones.

\end{FULL}
We describe in \begin{SHORT} the full version~\cite{uwram_full} \end{SHORT} two \uwram~algorithms for Shift-And that illustrate different techniques, noting that the \uwram~implementation of
Shift-Or is analogous. We obtain the following theorem:

\begin{restatable}[]{theorem}{shiftand}
\label{thm:shiftand}
Given a text $T$ of length $n$ and a pattern $P$ of length $m$, we can find the $occ$ occurrences of $P$ in $T$ in the \uwram~in time $O(nm/w^2+n/w+occ)$.
\end{restatable}

\begin{FULL}

\begin{algorithm}[t!]
\caption[UW-RAM Shift-And]{Shift-And($T,P,n=|T|,m=|P|,\Sigma$)}
\begin{algorithmic}[1]
			\STATE \COMMENT{Preprocessing}
			\FOR{each $\sigma \in \Sigma$}
					\STATE $Y[\sigma] \gets \vec{0}$
			\ENDFOR
			\FOR{$j=1$ to $m$}
					\STATE $Y[P[j]] \gets Y[P[j]]\ |\ (1>>(j-1))$
			\ENDFOR
			\STATE \COMMENT{Search}
			\STATE $V \gets \vec{0}$
			\STATE $C \gets 1>>(m-1)$
			\FOR{$i=1$ to $n$}
					\STATE $V=((V>>1)\ |\ 1)\ \& \ Y[T[i]]$
					\IF{$V\ \&\ C\ \neq 0$}
						\STATE report an occurrence at $i-m+1$
					\ENDIF
			\ENDFOR
\end{algorithmic}
\label{alg:shift_and1}
\end{algorithm}

\begin{algorithm}[t!]
\caption[UW-RAM Parallel Shift-And]{Parallel Shift-And($T,P,n=|T|,m=|P|,\Sigma$). 
For technical reasons, assume that $T[n+j]=\$$ for $j=1,\ldots,m-1$, with
  $\$\notin \Sigma$, and that $w\ge \log(n+m)$. In order to report
  matches at each step in time proportional to the number of matches
  (and not the number of blocks), we move directly to blocks with
  matching positions by using a function that for every word of length
  $w$ returns an array $A$ with the positions of set bits. For
  example, for $w=5$ and $x=01011$, $A=[1,3,4]$. We do this by table
  look up to a table with $(w/2)$-bit entries, whose space is
  $O(2^{w/2}w)$ words, which for $w=\log n$ is $O(\sqrt{n}\log
    n)$.}
\begin{algorithmic}[1]
			\STATE \COMMENT{Preprocessing}
			\FOR{each $\sigma \in \Sigma$}
					\STATE $Y[\sigma] \gets 0$ \COMMENT{$|Y[\sigma]|=w$}
			\ENDFOR
			\FOR{$j=1$ to $m$}
					\STATE $Y[P[j]] \gets Y[P[j]]\ |\ (1>>(j-1))$
			\ENDFOR
			\STATE $Y[\$]\gets 0$
			\STATE $V \gets \vec{0}$  
			\STATE $\mathrm{ONES} \gets \frac{2^{w^2}-1}{2^{w}-1}$ \COMMENT{$\mathrm{ONES}_j=1$ for all $j$}
			\STATE $C \gets \mathrm{ONES} >> (w-1)$ \COMMENT{$C_j=2^{w-1}$ for all $j$}
			\STATE \COMMENT{Search} 
			\STATE $n' \gets n/w$
			\STATE $\mathrm{POSNS}\gets \vec{0}$ \COMMENT{current positions in text}
			\FOR{$j=0$ to $w$} 
					\STATE $\mathrm{POSNS}\gets \mathrm{POSNS}\ |\ ((jn'+1)>>wj)$ 
			\ENDFOR
			
			\FOR{$i=1$ to $n'+m-1$}
					\STATE $V1 \gets (V>>1)\ |\ \mathrm{ONES}$
					\STATE $V2 \gets \mathrm{POSNS}$
					\STATE read\_content$(V2,T)$ \COMMENT{load characters in each position $(V2_j=T[\mathrm{POSNS}_j])$}
					\STATE read\_content$(V2,Y)$ \COMMENT{lookup masks in array $Y$ $(V2_j=Y[T[\mathrm{POSNS}_j]])$}
					\STATE $V\gets V1\ \&\ V2$
					\STATE $W \gets V\ \&\ C$ \COMMENT{check for matches at each block}
					\STATE $W \gets$ \comp$(W<<w-1)$
					\STATE $\mathrm{matches} \gets W_{0}$ \COMMENT{$\mathrm{matches}[j]=1$ if there was a match at block $j$}
					\STATE write\_word$(\mathrm{POSNS},\mathrm{matching\_positions})$ \COMMENT{write all current positions in array matching\_positions}
					\STATE $A \gets \mathrm{lookup}(\mathrm{matches})$ \COMMENT{position in $T$ of $k$-th matching block is at $\mathrm{matching\_positions}[A[k]]$}
					\FOR{$k=1$ to $|A|$}
						\STATE report match at $\mathrm{matching\_positions}[A[k]]$
					\ENDFOR
					\STATE $V \gets V\ \&\ \sim C$ \COMMENT{clear most significant bit in each block}
					\STATE $\mathrm{POSNS} \gets \mathrm{POSNS} + \mathrm{ONES}$ \COMMENT{update positions in $T$ ($\mathrm{POSNS}_j\le n+m-1$ for all $j$, thus there is no carry across blocks)}
			\ENDFOR
\end{algorithmic}
\label{alg:shift_and2}
\end{algorithm}

\subsubsection{$w^2$-bit Automaton} 
The straightforward way of taking advantage of the wide word when
implementing Shift-And is to use the entire wide word for bit
vectors. We first compute the mask array $Y[\sigma]$ for each
$\sigma\in \Sigma$ and store each $w^2$-bit vector in contiguous words
of memory starting at address $Y+\sigma$. Then the code of the \uwram~is essentially the same as the original code, replacing all references
to the array $Y$ with memory access operations for the wide word:
assuming $m\le w^2$, reading from and writing to $Y[\sigma]$ implemented
by read\_word$(W,Y+\sigma)$ and write\_word$(W,Y+\sigma)$, for
some word $W$. Otherwise, bit vectors are represented in $\lceil
m/w^2\rceil$ wide words (and stored in memory in $\lceil m/w^2\rceil
w$ words). The rest of the operations are done on registers, and
constants are part of the precomputation. The pseudocode for this
algorithm is shown in Algorithm~\ref{alg:shift_and1}, which assumes
$m\le w^2$ and is based on the pseudocode for Shift-And given
in~\citep[Chapter~2.2.2]{NRbook02}. Since we can now update $\vec{v}$ in
$O(m/w^2+1)$ time, the running time of Algorithm~\ref{alg:shift_and1}
is $O(nm/w^2+n)$. Thus, compared to the original algorithm, the \uwram~algorithm achieves a speedup of $w$ when $m\ge w^2$, and a speedup of
$\lceil m/w\rceil$ otherwise (no speedup is achieved for $m\le w$).

\begin{lemma}
When implemented in the \uwram, the Shift-And and Shift-Or algorithms for searching a pattern of length $m$ in a text of length $n$ have a running time of $O(nm/w^2+n)$, achieving a $w$-fold speedup over word-\ram~implementations when $m\ge w^2$.
\end{lemma}

\subsubsection{$w$-bit Parallel Automata} 
Another way of using the wide word to speedup the Shift-And algorithm
is to take advantage of the parallel memory access operations of the
\uwram~to perform $w$ parallel searches on disjoint portions of the
text. This is done by using each block of a wide word to represent the
automaton in each search: block $j$ is used to search $P$ in
$T[jn/w..(j+1)n/w-1]$, for $0\le j\le w-1$ (we assume $w$ divides
$n$). Since the operations involved in updating the automata are the
same across blocks, an update to all $w$ automata can be done with a
constant number of single wide word operations. All bit vectors of the
precomputed table $Y$ are now again $w$-bit long, as in the original
algorithm. In each step of the search, $w$ entries of $Y$ are read in
parallel to each block according to the current character in $T$ in
the search in each portion. The pseudocode for this procedure is shown
in Algorithm~\ref{alg:shift_and2}. The code assumes $m\le w$, though
it is straightforward to modify it for the $m>w$ case. The running
time of this algorithm is now $O(nm/w^2+n/w+occ)$, where $occ$ is the
number of occurrences found. This is asymptotically faster than the first
version above, and it leads to Theorem~\ref{thm:shiftand}.

\end{FULL}

\begin{SHORT}
\paragraph{\textbf{Boyer-Moore-Horspool}}
\label{sec:bmh}
\bmh~\cite{horspool80} keeps a sliding window of length $m$ over the
\end{SHORT}
\begin{FULL}
\subsection{Boyer-Moore-Horspool}
\bmh~\cite{horspool80} keeps a sliding window of length $m$ over the
\label{sec:bmh}
\end{FULL}
text $T$ and searches backwards in the window for matching suffixes of
both the window and the pattern. \begin{FULL}More specifically, for a window
$T[i..i+m-1]$, the algorithm checks if $T[i+j-1]=P[j]$ starting with
$j=m$ and decrementing $j$ until either $j=0$ (there is a match) or a
mismatch is found. Either way, the window is then shifted so that
$T[i+m-1]$ is aligned with the last occurrence of this character in
$P$ (not counting $P[m]$).\end{FULL} The worst case running time of \bmh is
$O(nm)$ (when the entire window is checked for all window positions)
but on average the window can be shifted by more than one character,
making the running time $O(n)$~\cite{BaezaYates92}. 
In the \uwram, we can take
advantage of the wide word to make several character comparisons in
parallel, thus achieving a $w$-fold speedup over the worst case behaviour of \bmh. 
\begin{SHORT} Full details are described in~\cite{uwram_full}.\looseness=-1\end{SHORT}
\begin{FULL}
A recent SIMD-based implementation of \bmh using SSE4.2 on Intel i5 and Xeon processors~\cite{ladra12} is evidence of the practicality of this approach.
\end{FULL}

\begin{FULL}
%
First, we divide each wide word in $f$-bit fields so
that each field contains one character, thus $f=\lceil\log \sigma\rceil$. At
each position of the window, we do a field-wise comparison between a
wide word containing the characters of the text and one containing the
characters of the pattern. We do this simply by subtracting both
words. Since we only care if all symbols in the words match, we only
need to check if the result is zero, without having to worry about
carries crossing fields (and hence we do not need a test bit). We
shift the window to the next position if the result is not zero. Note
that this check can be done in constant time, and it is quite simple as
we do not need to identify where there was a mismatch. Thus in each
window we can compare up to $w^2/f$ symbols in parallel, and hence the
running time in the worst case becomes $O(mn\log \sigma/w^2+1)$. We show
the pseudocode in Algorithm~\ref{alg:bmh} which, again, is based on
the pseudocode of this algorithm presented
in~\citep[Chapter~2.3.2]{NRbook02}. Note that for a given input the
distance of the shifts is exactly the same as in the original version
of the algorithm, and therefore the average running time remains the
same. Note as well that the average running time can be reduced by
using each block to search in disjoint parts of the text at the
expense of increasing the worst case time to $O(mn\log \sigma/w+1)$ due
to the reduction in the number of characters that can be compared
simultaneously.
\end{FULL}

\begin{restatable}[]{theorem}{bmh}
Given $T$ of length $n$ and $P$ of length $m$ over an alphabet of size $\sigma$, we can find the occurrences of $P$ in $T$ with a \uwram~implementation of BMH in $O(mn\log \sigma/w^2+1)$ time in the worst-case and $O(n)$ time on average.\looseness=-1
\end{restatable}

\begin{FULL}
\begin{algorithm}[t!]
\caption[UW-RAM BMH]{\label{alg:bmh}BMH($T,P,n=|T|,m=|P|,\Sigma$). 
For simplicity,
  we assume that $w$ divides $m\log\sigma$. We assume also that $T$
  and $P$ are represented with $\log\sigma$ bits per symbol. We still
  use $T[i]$ to denote one character, which can be easily obtained
  from the packed representation in constant time (the same applies to
  the actual address of starting characters of substrings).}
\begin{algorithmic}[1]
			\STATE \COMMENT{Preprocessing}
			\FOR{each $\sigma \in \Sigma$}
					\STATE $\mathrm{jump}[\sigma] \gets m$
			\ENDFOR
			\FOR{$j=1$ to $m-1$}
					\STATE $\mathrm{jump}[P[j]] \gets m-j$
			\ENDFOR
			\STATE $m'\gets w^2/\log\sigma$ \COMMENT{characters per wide word}
			\STATE \COMMENT{Search}
			\STATE $i=0$
			\WHILE{$i\le n-m$}
			\STATE $k\gets m'/m$ \COMMENT{number of window segment}
			\WHILE{$k>0$}
				\STATE $W \gets T[i+(k-1)m'+1..i+km']$ \COMMENT{$W$ contains the substring of $T$ of $k$-th window segment}
				\STATE $V \gets P[(k-1)m'+1..km']$ \COMMENT{$V$ contains the substring of $P$ of $k$-th window segment}
			\IF{$W-V \ne 0$}
				\STATE break
			\ELSIF{$k = 1$}
				\STATE report occurrence at $i+1$
			\ENDIF
			\STATE $k\gets k-1$
			\ENDWHILE
			\STATE $i \gets i+\mathrm{jump}[T[i+m]]$
			\ENDWHILE
\end{algorithmic}
\end{algorithm}


\end{FULL}

\section{Conclusions}
\label{sec:conclusions}

We introduced the Ultra-Wide Word architecture and model and
showed that several classes of algorithms can be readily implemented
in this model to achieve a speedup of $\Omega(\log n)$ over traditional word-\ram~algorithms.  The examples we describe 
already show the
potential of this model to enable parallel implementations of existing
algorithms with speedups comparable to those of multi-core
computations.  
We believe that this architecture could also serve to
simplify many existing word-\ram~algorithms that in practice do not
perform well due to large constant factors. We conjecture as well
that this model will lead to new efficient algorithms and data
structures that can sidestep existing lower bounds.\looseness=-1

\small

\bibliographystyle{splncsnat}

\bibliography{succinct_short,parallel_short,string_search_short}
\normalsize

\begin{FULL}
\appendix

\section*{Appendix}

\section{UW-RAM Subroutines}
\label{app:subroutines}

\paragraph{\textbf{Comparators}}
Many word-\ram~algorithms perform operations on pairs of elements in
parallel by packing these elements in \emph{fields} within one
word. It is useful to be able to do fieldwise comparisons between two
words. Suppose that a word (either regular or wide)
is divided in $f$-bit
fields, with each field representing an $(f-1)$-bit number. Let $G$
and $F$ be two such words and let $F_i$ and $G_i$ denote the contents
of the $i$-th field in $F$ and $G$, respectively. Let us assume that we want
to identify all $F_i$ such that $F_i\ge G_i$. Fieldwise comparisons can be done by setting the most significant bit
of each field in $F$ as a test bit and computing $H=F-G$. The most
significant bit of the $i$-th field in $H$ will be 1 if and only if
$F_i\ge G_i$~\cite{hagerup98}. Now, if we want to operate only on
the values of $F$ that are greater than or equal to their
corresponding values in $G$, we can mask away the rest of the values as follows. We first mask away all but the test bits
in $H$. Then, a mask $M$ with ones in all bits of the relevant fields and
zeros everywhere else (including test bits) can be obtained by
computing $M=H-(H<<(f-1))$. The result of $(M \ \& \ F)$ contains then
only the values of fields that pass the test~\cite{hagerup98}. Clearly
this operation takes constant time, and it can be easily adapted to
other standard comparisons. We shall
assume that direct comparisons as well as operations that build on
these (such as taking the fieldwise maximum between two words) are
available and take constant time~\cite{hagerup98}.

\end{FULL}

\end{document}